\newif\ifarxiv
\arxivtrue 
\documentclass[sigconf]{acmart}

\usepackage[shortlabels]{enumitem} 

\newcommand{\p}[1]{\left( #1 \right)}

\newcommand{\cd}[0]{\cdot}

\newlist{todolist}{itemize}{2}
\setlist[todolist]{label=$\square$}
\usepackage{pifont}

\usepackage{amsmath}               
\usepackage{amsthm}                
\usepackage{thmtools,thm-restate}
\usepackage[autostyle=true,german=quotes]{csquotes}
\usepackage{multirow}

\usepackage{array}
\usepackage{mathtools}
\newcolumntype{C}[1]{>{\centering\let\newline\\\arraybackslash\hspace{0pt}}m{#1}}

\newtheorem{lemma}{Lemma}

\newtheorem{definition}{Definition}

\newcommand{\nplayer}[0]{\ensuremath{M}}
\newcommand{\gendist}[0]{\ensuremath{\Theta}}
\newcommand{\mean}[0]{\ensuremath{\theta}}
\newcommand{\mue}[0]{\ensuremath{\mu_e}}
\newcommand{\var}[0]{\ensuremath{\sigma^2}}
\newcommand{\sampledist}[0]{\ensuremath{\mathcal{D}}}

\newcommand{\ndraw}[0]{\ensuremath{n}}
\newcommand{\err}[0]{\ensuremath{\epsilon^2}}

\newcommand{\expparam}[1]{\ensuremath{\mathbb{E}_{(\mean_{#1}, \err_{#1}) \sim \gendist}}}
\newcommand{\total}[0]{\ensuremath{N}}

\newcommand{\Ymf}[0]{\ensuremath{Y}}

\newcommand{\vm}[0]{\ensuremath{v}}
\newcommand{\alone}[0]{\ensuremath{\pi_l}}

\newcommand{\s}[0]{\ensuremath{s}}
\newcommand{\el}[0]{\ensuremath{\ell}}
\newcommand{\ns}[0]{\ensuremath{n_{\s}}}
\newcommand{\nlv}[0]{\ensuremath{n_{\el}}}

\newcommand{\si}[0]{\ensuremath{S}}
\newcommand{\li}[0]{\ensuremath{L}}
\newcommand{\col}[0]{\ensuremath{C}}

\newcommand{\partition}[0]{\ensuremath{\Pi}}

\AtBeginDocument{%
  \providecommand\BibTeX{{%
    \normalfont B\kern-0.5em{\scshape i\kern-0.25em b}\kern-0.8em\TeX}}}

\copyrightyear{2023}
\acmYear{2023}
\setcopyright{acmlicensed}\acmConference[WWW '23]{Proceedings of the ACM Web Conference 2023}{May 1--5, 2023}{Austin, TX, USA}
\acmBooktitle{Proceedings of the ACM Web Conference 2023 (WWW '23), May 1--5, 2023, Austin, TX, USA}
\acmPrice{15.00}
\acmDOI{10.1145/3543507.3583483}
\acmISBN{978-1-4503-9416-1/23/04}

\begin{document}

\title{Fairness in model-sharing games}

\author{Kate Donahue}
\email{kdonahue@cs.cornell.edu}
\affiliation{
  \institution{Department of Computer Science, Cornell University}
  \country{USA}
}

\author{Jon Kleinberg}
\email{kleinberg@cornell.edu}
\affiliation{
  \institution{Departments of Information and Computer Science, Cornell University}
  \country{USA}
}

\renewcommand{\shortauthors}{Donahue and Kleinberg}

\begin{abstract}
In many real-world situations, data is distributed across multiple self-interested agents. These agents can collaborate to build a machine learning model based on data from multiple agents, potentially reducing the error each experiences. However, sharing models in this way raises questions of fairness: to what extent can the error experienced by one agent be significantly lower than the error experienced by another agent in the same coalition? In this work, we consider two notions of fairness that each may be appropriate in different circumstances: \emph{egalitarian fairness} (which aims to bound how dissimilar error rates can be) and \emph{proportional fairness} (which aims to reward players for contributing more data). We similarly consider two common methods of model aggregation, one where a single model is created for all agents (uniform), and one where an individualized model is created for each agent. For egalitarian fairness, we obtain a tight multiplicative bound on how widely error rates can diverge between agents collaborating (which holds for both aggregation methods). For proportional fairness, we show that the individualized aggregation method always gives a small player error that is upper bounded by proportionality. For uniform aggregation, we show that this upper bound is guaranteed for any individually rational coalition (where no player wishes to leave to do local learning). 
\end{abstract}






\maketitle
\section{Introduction}

In traditional machine learning, an agent (such as a hospital, school, or company) trains a model based on data it has locally. In many cases, the agent may only have access to a limited amount of local data, which may mean that any model created with it will have relatively high error. However, in many real-world scenarios, multiple agents all are trying to solve similar problems using related local data sets. For example, consider a hospital aiming to train a model predicting patient outcomes. \emph{Local learning} would rely on the hospital's own data from observations of its patients. However, it may be the case that multiple hospitals in the same area all similarly have local patient data related to this problem. This raises the question - could it be possible to leverage this related data to help build a better model? 

Multiple research methods address this very problem. For example, data cooperatives and data trusts allow individuals to pool their data, often through the use of a trusted intermediary, in order to safely produce higher performance for those participating \cite{van2017potentials, LOMOTEY2022100075, hardjono2019data, hafen2014health}. Another method is federated learning, in which a central agent combines learned parameters (rather than the raw data itself) from multiple federating agents \cite{Li_2020, kairouz2019advances}. More broadly, edge computing studies ways of leveraging data distributed among multiple nodes with communication limitations \cite{10.1145/3464419}. All of these methods have been the subject of multiple game theoretical analyses \cite{blum2021one, hasan2021incentive, grimberg2021optimal,donahue2020model, donahue2021optimality, tu2021incentive, LOZANO2012558, 9144224, karimireddy2022mechanisms}. We will refer to situations where multiple agents cooperatively share data or model parameters as \emph{model-sharing games}, using terminology from  our prior work \cite{donahue2020model}.

In general, model sharing can provide benefits to participating agents by allowing them to produce models with lower error. However, this may itself produce additional problems, especially around fairness. Often different agents using the same model will see different levels of error, potentially as a result of differences in how much data each agent contributed. While this disparity might be undesirable, a \enquote{fair} solution can have multiple definitions, reflecting different values and goals of the particular situation. 

For example, consider a case with two hospitals $A$ and $B$, where $A$ has fewer relevant data points than $B$. It might be the case that $A$ has less data because it is a smaller hospital with fewer resources. In this case, a reasonable fairness goal might be for the model sharing solution to rectify this inequity, producing errors for $A$ and $B$ that are approximately equal. The goal of having error rates that are the same or similar is reflected in \emph{egalitarian fairness}. However, a different scenario might call for a different fairness goal. Consider the case where hospitals $A$ and $B$ have similar access to resources, but hospital $A$ has less data because it has devoted less effort to collecting data. In this case, a reasonable fairness objective might be to reward player $B$ for its higher contributions. One specific approach might be to have each agent's error rate be inversely proportional to the amount of data it contributes: this motivates \emph{proportional fairness}. While egalitarian and proportional fairness can't be simultaneously achieved, each definition reflects goals that could be reasonable for different situations. Multiple papers separately discuss notions related to egalitarian or proportional fairness, in data cooperatives \cite{8761948} or in federated learning \cite{mohri2019agnostic,li2019fair, Lyu2020CollaboratfiveFI, Zhang2020HierarchicallyFF, li2021ditto, fan2022improving}.

{\textbf{The present work:}}\\
In this work, we provide the first theoretical analysis and comparison of egalitarian and proportional fairness. Our paper is primarily analytical rather than prescriptive: we describe patterns of fairness that are and aren't satisfied for commonly-used methods of model aggregation. In order for this analysis to be tractable theoretically, we build on the simplified model of model-sharing games we first proposed in \citet{donahue2020model}, which we summarize later in this work for convenience of the reader.
\begin{itemize}
    \item For \textbf{egalitarian fairness}, we give a \emph{tight multiplicative bound on how dissimilar the error rates of collaborating agents can be}. This bound varies based on the number of samples the largest agent contributes, as well as a property of the learning task related to the ratio of noise and bias. We show that this bound holds for a wide range of model-sharing methods satisfying a few straightforward properties. 
    \item For \textbf{proportional fairness}, our results vary based on the model creation method. We identify {\em proportionality} as a natural baseline that emerges from the analysis: Proportionality would mean that if agent $A$ has $c$ times the number of samples as agent $B$, it would obtain $1/c$ times as much error; this is the property that would hold, for example, for the error rates if none of the agents shared any data. For fine-grained learning, where a personalized model is made for each agent, we show any coalition formation results in the sub-proportional error: that is, small players end up with error rates that are lower than proportionality would suggest. By contrast, for uniform model-sharing, where a single model is created for the entire coalition, we show that in general there exist coalitions that could display error rates that are arbitrarily disconnected from proportionality. That is, the error an agent experiences, compared to the error another agent experiences, could be either more or less than what the ratio of the size of their datasets suggests. However, we show that for uniform model-sharing, any coalition where small players experience error that is larger than proportionality fails to be \emph{individually rational}: at least one agent could achieve lower error by leaving the coalition for local learning. 
\end{itemize}

Taken together, these results indicate that smaller players (those who contribute fewer samples) in general obtain higher errors in model-sharing games. However, both fairness notions show that this disparity in performance is bounded: either by the multiplicative bound in egalitarian fairness, or by the model aggregation method in proportional fairness. Again, the focus of our work is \emph{descriptive}: while these fairness results may or may not be desirable in different situations, our goal is to describe certain guarantees that arise from different model-sharing methods. 

\section{Related work and fairness definitions}\label{sec:fairdef}

Fairness in model-sharing is an especially rich topic: naturally, no one set of definitions is going to resolve the complex questions it raises. Much of the literature focuses on fairness in federated learning: for a survey of fairness issues here and methods to address them, refer to \cite{shi2021survey}. In this work, our goal is to theoretically evaluate two dominant notions of fairness in collaboratively building models\footnote{Naturally, some fairness notions fall outside of these two notions, such as \cite{chaudhury2022fairness}, which proposes and analyzes definitions fairness based on core stability of coalitions.  }.
 
\subsection{Egalitarian fairness}
We adopt the terminology from \citet{Xu2020TowardsBA} calling the first \enquote{egalitarian}. Egalitarian notions of fairness roughly revolve around ensuring agents have levels of error that are all roughly comparable. For example, \citet{mohri2019agnostic} works to minimize the error of the worst-off player, while \citet{Du2020FairnessawareAF} considers a similar notion with an explicit fairness penalty. By contrast, \citet{li2019fair, li2021ditto} has a more flexible goal of minimizing the spread of error rates. \citet{abay2020mitigating} analyzes multiple causes of unfairness in federated learning and potential solutions. Egatliarian fairness has also been explored in data cooperatives \cite{8761948}. Other papers that consider notions related to egalitarian fairness include \cite{zeng2021improving, zhang2021unified, papadaki2021federating, kanaparthy2021fair}. 

The motivation for this notion of fairness is that agents may differ in how much data they have for reasons beyond their control: persistent inequities in resources, for example. The goal of model-sharing, in this view, is to try and ameliorate these injustices by bringing error rates more close to each other. We formalize this idea in Definition \ref{def:egal}.  
\begin{definition}\label{def:egal}
If a player $i$ in a coalition $\col$ experiences error $err_i(\col)$, then the coalition $\col$ satisfies $\lambda$-egalitarian fairness if for some constant $\lambda \geq 1$,
$$\frac{err_i(\col)}{err_j(\col)} \leq \lambda \ \forall i, j \in \col$$
\end{definition}

\subsection{Proportional (collaborative) fairness}
By contrast, proportional notions of fairness generally focus on ensuring that different agents are rewarded for contributing more to the overall model \cite{fan2022improving}. For example, \citet{Lyu2020CollaborativeFI,Xu2020TowardsBA} both use gradient-based methods to detect players contributing useful model updates (as opposed to noise) and reward them with more powerful variants of the model. \citet{Song2019ProfitAF} takes a similar approach with the goal of allocating profit. \citet{Zhang2020HierarchicallyFF} proposes having a trusted party track the amount of effort each player is contributing, which will then be used to reward or penalize agents with access to different models. 

Under the view of proportional fairness, the number of samples a player contributes to the shared model is under their control, so we wish to reward players who contribute more. In Definition \ref{def:prop}, we formalize the idea of proportional fairness by describing how error rates relate to the number of samples that a player contributes.

\begin{definition}\label{def:prop}
A coalition $\col$ with satisfies $\alpha$-proportional fairness if for all pairs of players $i$ and $j$, with $\ndraw_i \leq \ndraw_j$, some $\alpha >0$,
$$err_i(\col) \leq \alpha \cd  \frac{\ndraw_j}{\ndraw_i}\cd err_j(\col) \ \forall i, j \in \col \ \ndraw_i \leq \ndraw_j$$
\end{definition}

Note that $\alpha<1$ means that smaller players get \emph{lower} error than proportionality would suggest, while $\alpha>1$ means that smaller players get \emph{higher} error than proportional. Either case may be desirable under different situations: we will restrict our attention to describing when either is possible to achieve.

\section{Model and assumptions}\label{sec:modelassump}

\subsection{Model sharing methods}
Next, we will present the framework we will use to analyze model-sharing. We say that there are $\nplayer$ total agents (sometimes referred to as players). Each agent $i \in [\nplayer]$ has $\ndraw_i$ data points, labeled by its true labeling distribution $g(\mean_i)$, where $\mean_i$ are their true local parameters and $g(\cd)$ is some labeling function (e.g. mean estimation, linear regression, deep learning). A player's goal is to learn a model $g(\hat \mean_i)$ with low expected error on its own distribution. If a player opts for local learning, then it uses its local estimate of these parameters $\hat \mean_i$ to predict future data points.  If a set of players $\col$ are building a model together, we say that they are in a \emph{coalition}  or \emph{cluster} together. An \emph{aggregation method} describes how a coalition aggregates information from participating agents to. For example, in federated learning, a federating method describes how players combine their learned parameters into federating models. The most common federating method is the weighted average method in Equation \eqref{eq:avged}, often called \enquote{vanilla} or \enquote{uniform} federated learning (\cite{mcmahan2016communicationefficient}): 
\begin{equation}\label{eq:avged}
 \hat \mean_\col = \frac{1}{\sum_{i \in \col} \ndraw_i} \cd \sum_{i \in \col} \ndraw_i \cd \hat \mean_i  
\end{equation}
Alternative ways of federation might involve customizing the model for individuals, as in domain adaptation. For example, \cite{donahue2020model} models three methods of federation: vanilla, as well as two models of domain adaptation, fine-grained and coarse-grained. In fine-grained federation, a personalized model is made for each player by weighting the contributions from every other player:
\begin{equation}\label{eq:fineavged}
 \hat{\mean}_j^{\vm} = \sum_{i=1}^{\nplayer}\vm_{ji}\cd \hat{\mean}_i 
\end{equation}
with $\sum_{i=1}^{\nplayer}\vm_{ji} = 1$. \enquote{Optimal} fine-grained federation occurs when these $\vm_{ji}$ weights are set so as to minimize error for each player and is the case we will consider throughout the rest of this paper. 

A player $i$ obtains error $err_i(\col)^m$, where $m$ denotes the federating method being used. Note that in general $err_i(\col)^m \ne err_j(\col)^m$: players in the same coalition may experience different error rates. For example, if player $i$ has more samples than player $j$, then $\hat \mean_\col$ as calculated in vanilla federation in Equation \eqref{eq:avged} will be weighted more towards player $i$, meaning that player $i$ will have lower expected error than $j$. While these aggregation methods are motivated by federated learning, they could also be used to model combining methods in other settings, such as data coalitions. 

\subsection{Theoretical model from \citet{donahue2020model}}\label{sec:existingmodel}

All of the previous definitions are satisfied for any model-sharing scenario. However, in order for our theoretical fairness analysis to be feasible, we require a model that gives a functional form for $err_i(\col)$: the error each player experiences in a coalition. Moreover, we require a model that gives exact errors for each player, rather than bounds: these are needed in order to be able to analyze the ratio of the errors two players experience, for example.  

We opt to use the model developed in our prior work \citet{donahue2020model}, which produces the closed-form error value seen in Lemma \ref{lem:error} below. While we work within this model, we emphasize that \cite{donahue2020model} asked different questions from this paper's focus: our prior work focused on developing the federated learning model and analyzing the stability of federating coalitions, while our current work analyzes optimality and Price of Anarchy. (Separately, in \cite{donahue2021optimality} we similarly build on \cite{donahue2020model}, but with the orthogonal goal of analyzing the \emph{optimality} of federated learning: which arrangements minimize total error.)

\begin{lemma}[Lemma 4.2, from \cite{donahue2020model}]\label{lem:error}
Consider a mean estimation task as follows: player $j$ is trying to learn its true mean $\mean_j$. It has access to $\ndraw_j$ samples drawn i.i.d. $\Ymf \sim \sampledist_j(\mean_j, \err_j)$, a distribution with mean $\mean_j$ and variance $\err_j$. Given a population of players, each has parameters $(\mean_j, \err_j) \sim \Theta$ taken from some distribution $\Theta$\footnote{Note that this does \emph{not} imply that each player draws their parameters i.i.d, but rather that there exists some distribution of true model parameters among players. For example, $\Theta$ might be an (unknown) distribution of true coronavirus rates among towns, reflecting variability in true rates. By contrast, $\sampledist_j(\mean_j, \err_j)$ would be the number of coronavirus cases in a particular town per day, which reflects noise in data collection. }  A coalition $\col$ using uniform federation produces a single model based on the weighted average of local means (Eq. \eqref{eq:avged}). Then, the expected mean squared error player $j$ experiences in coalition $\col$ using vanilla federation is:
\begin{equation}\label{eq:err}
err_j(\col) = \frac{\mue}{\sum_{i \in \col} \ndraw_i} + \var \cd \frac{\sum_{i \in \col, i \ne j}\ndraw_i^2 + \p{\sum_{i\in \col, i \ne j}\ndraw_i}^2}{\p{\sum_{i \in \col}\ndraw_i}^2}
\end{equation}
where $\mue = \expparam{i}[\err_i]$ (the average noise in data sampling) and $\var = Var(\mean)$ (the average dissimilarity or bias between the true means of players). 
\end{lemma}

\begin{lemma}[Lemmas 4.1, 7.1, from \cite{donahue2020model}]\label{lem:errorfine}
Consider the same setting as in Lemma \ref{lem:error}, but where the coalition $\col$ was instead using personalized (fine-grained) federation. Then, the expected mean squared error player $j$ experiences in the coalition with fine-grained federation is given by: 
\begin{equation}\label{eq:errfine}
\mue\sum_{i=1}^{\nplayer}\vm_{ji}^2\cd \frac{1}{\ndraw_i} + \p{\sum_{i\ne j}\vm_{ji}^2 + \p{\sum_{i\ne j}\vm_{ji}}^2}\cd \var
\end{equation}
\enquote{Optimal} fine-grained federation occurs when each $\vm_{ji}$ parameter is set so as to minimize error for player $j$. Note that every player can simultaneously minimize their errors: player $k \ne j$ can set their parameters $\{\vm_{ki}\}$ separately from $\{\vm_{ji}\}$. The values of $\vm_{ji}$ that minimize error are given by: 
$$\vm_{jj} = \frac{1 + \var \sum_{i\ne j}\frac{1}{V_i}}{1 + V_j\sum_{i\ne j}\frac{1}{V_i}}\quad \vm_{jk} = \frac{1}{V_k}\cd \frac{V_j-\var}{1 + V_j \sum_{i\ne j}\frac{1}{V_i}} \quad k\ne j$$
for $V_i = \var + \frac{\mue}{\ndraw_i}$. 
\end{lemma}
Note that \cite{donahue2020model} also analyzes a $d$-dimensional linear regression game and in particular show that the cost functions for this game are equivalent to the costs for mean estimation, so long as the number of samples $n$ is much larger than the dimension of the problem $d$. As a result, all of our findings directly extend from mean estimation to linear regression. Additionally, theoretical and empirical work (\cite{hestness2017deep, rolf2021representation}) indicates that power law rules for error curves are commonly used, indicating that functional forms similar to those proposed in \cite{donahue2020model} may be much more broadly applicable. 

We use some of the same modeling assumptions as \cite{donahue2020model}. For example, we assume that each player has a goal of obtaining a model with low expected test error on its personal distribution. For notation, we use $\col$ to refer to a coalition of federating agents and $\partition$ to refer to a collection of coalitions that partitions the $\nplayer$ agents. We will use $\total_{\col}$ to refer to the total number of samples present in coalition $\col$: $\total_{\col} = \sum_{i \in \col}\ndraw_i$. We will refer to $\frac{\mue}{\var}$ as the \emph{noise/bias ratio}, given that it reflects the level of noise in a learning problem, as compared to the average amount of bias between the true parameters of different players. We also assume that each agent is always \enquote{available} to participate in model-sharing, meaning that our setting is more close to cross-silo federated learning, as opposed to cross-device \cite{wu2022motley}. 

Finally, it is worth emphasizing key differences between this current work and \cite{donahue2020model}. The focus of \cite{donahue2020model} is defining a theoretical model of federated learning and taking a game-theoretic analysis of the stability of such an arrangement. While we use the theoretical model we developed in \cite{donahue2020model}, in this work we analyze completely distinct questions: what are reasonable definitions of \enquote{fairness} in federated learning? When are certain fairness notions guaranteed to hold?  Additionally, this paper is in some ways more general: while some key results in \cite{donahue2020model} only allow players to have two different numbers of samples (\enquote{small} or \enquote{large}), every result in this work holds for arbitrarily many players with arbitrarily many different numbers of samples.

\subsection{Research ethics and social impact}\label{sec:ethics}
In Section \ref{sec:fairdef} we began our discussion on fairness and ethics, which we will continue here. This work is primarily theoretical, but touches on questions with multiple possible societal implications. As we discussed previously, there are many reasonable definitions of fairness in model-sharing games (each of which could impact participants differently). In this work, we narrow our focus to two of the most commonly used and natural definitions, but it is important to note that this by no means is the last word on fairness in this area. For example, our prior work \cite{donahue2020model} derived a cost function for model-sharing games assuming that each agent wishes to minimize their mean-squared error. However, it is possible that agents might care about metrics besides \emph{expected} error. For example, \cite{coopervariance} explores how disparities in \emph{variance} of error rates between groups can cause significant harms, including outcomes that feel arbitrary. Beyond this, fairness could mean something besides error rates, such as equality of access to opportunities or access to recourse.

Besides fairness, model and data sharing also has privacy concerns. While federated learning involves sharing model parameters, rather than raw data, under some circumstances it still may be possible to learn information about the data distribution of specific players \cite{mothukuri2021survey, boenisch2021curious}. Some papers have explored the game-theoretic implications of privacy concerns in federated or data-sharing settings \cite{cummings2022optimal, bie2022private, li2023differentially, 9996132}, as well as \cite{kang2023fair}, which builds on the model-sharing games framework from \cite{donahue2020model}. Finally, one core assumption of this paper is that the problem at hand will be improved by a machine learning model with lower error. In many situations, a different approach might be more suitable, such as an algorithm that is more explainable or a non-machine learning solution. 

\section{Egalitarian fairness}\label{sec:egal}

\subsection{Motivation}

In this section, we will analyze \emph{egalitarian fairness}, where a typical goal is to ensure that all players have similar levels of error. This may be motivated by the fact that players may differ in their circumstances (differing amounts of data) for reasons out of their control. For example, imagine a low-resource hospital that is collaborating with a high-resource hospital. The low-resource hospital may have access to less local data for historical and/or continuing reasons of discrimination and unequal access to technology and other resources. This motivation would suggest that model-sharing algorithms should seek to close the gap the low resource and high resource hospital, making their error as close to equal as possible. In some cases, it may even be desirable for the low-resource (small) hospital to have lower error than the high-resource hospital, potentially as a way to begin to compensate for its lack of resources. In this section, we give theoretical bounds for what types of egalitarian fairness levels can be obtained. 

\subsection{Tight, widely-applicable upper bound}

Theorem \ref{thrm:egalbound} shows that there exists a tight multiplicative bound for egalitarian fairness: that is, the error ratio between small and large players is upper bounded. This holds for any \emph{modular function}, where definition \ref{def:modular}, below, gives a formal definition of what this means. Later, we will directly show that both vanilla and fine-grained federation are modular. However, because the definition of \enquote{modular} encompasses multiple desirable properties of model-aggregation methods, we expect that a wide range of model-aggregation methods would satisfy them. Thus, Theorem \ref{thrm:egalbound} is best viewed as a relatively broad theorem potentially providing insight into multiple different aggregation scenarios.  

\begin{restatable}{theorem}{egalbound}
\label{thrm:egalbound}
Any \enquote{modular} model-sharing method (as defined in Definition \ref{def:modular}) satisfies the egalitarian error ratio bound: 
$$\frac{err_i(\col)}{err_j(\col)} \leq 2 \cd c + 1 \ \forall i, j \in \col$$
so long as the largest player size is $ \leq c \cd \frac{\mue}{\var}$ (for $\frac{\mue}{\var}$ noise/bias ratio, for any $c > 0$). Moreover, this bound is tight (up to an additive factor of $\epsilon$).
\end{restatable}

\begin{definition}\label{def:modular}
Consider any coalition $\col$ with players $\s, \el \in \col$ such that $\ns \leq \nlv$. Then, a model-sharing method is called modular if it satisfies the following five properties: 

\textbf{Property 1:} The large player always has lower error than the small player: 
$$\frac{err_s(\col)}{err_l(\col)} \geq 1$$
\textbf{Property 2:}
The worst-case situation for the error ratio (the ratio of the small player's error to the large player's error) is always in the two-player case. That is, 
$$\frac{err_s\p{\col}}{err_l\p{\col}} \leq \frac{err_s\p{\{\ns, \nlv\}}}{err_l\p{\{\ns, \nlv}\}} $$
\textbf{Property 3:} The error ratio increases as the large player gets more samples: the small player's error either increases or else decreases more slowly than the large player's error decreases. That is, 
$$\frac{\partial}{\partial\nlv}\frac{err_s\p{\{\ns, \nlv\}}}{err_l\p{\{\ns, \nlv}\}} \geq 0$$
\textbf{Property 4:} The error ratio decreases as the small player gets more samples: the large player's error either increases or else decreases more slowly than the small player's error decreases. That is, 
$$\frac{\partial}{\partial\ns}\frac{err_s\p{\{\ns, \nlv\}}}{err_l\p{\{\ns, \nlv}\}} \leq 0$$
\textbf{Property 5:} As $\frac{\ns}{\nlv} \rightarrow 0$ (the large player has many more samples than the small player), the error ratio converges to the following fraction: 
$$\lim_{\frac{\ns}{\nlv} \rightarrow 0}\frac{err_s\p{\{\ns, \nlv}\}}{err_l\p{\{\ns, \nlv}\}} =\frac{\frac{\mue}{\nlv} + 2 \var}{\frac{\mue}{\nlv}}$$
This fraction reflects the fact that the large player's error approaches that of local learning, while the small player's error approaches that of simply using the model of the large player without contributing any data. 
\end{definition}

First, it may be useful to discuss the properties in Definition \ref{def:modular}. Property 1 is simply stating that, for any pair of players both federating in the same coalition, the smaller of the two always has higher error. The remaining four properties contribute to the upper bound (and showing that this bound is tight). First, Property 2 states that the error ratio between any two players is highest when they are federating in a two-player coalition (without any other players). This property is especially helpful because the two-player coalition is much simpler to analyze. Properties 3 and 4 consider this two-player coalition, showing that the error ratio is highest when the large player has many samples and the small player has very few. Finally, Property 5 describes the limit of the error ratio in the case where the large player has many more samples than the small player, stating that it must go to a fixed ratio depending on the noise-bias ratio, as well as the number of samples the large player has. Roughly speaking, the denominator of this term ($\frac{\mue}{\nlv}$) approximates the error that the large player $\ell$ obtains while federating alone, while the numerator reflects the error that the small player obtains while using the model created almost entirely by the large player ($\frac{\mue}{\nlv}$, plus a bias term of $2 \var$). 

Given these properties, it becomes straightforward to prove Theorem \ref{thrm:egalbound}: 
\begin{proof}[Proof of Theorem \ref{thrm:egalbound}]

Proving the $2 c +1$ bound is given by applying the four properties: 
    $$\frac{err_s\p{C}}{err_l\p{C}} \leq \frac{err_s\p{\{\ns, \nlv\}}}{err_l\p{\{\ns, \nlv}\}} \leq \frac{\frac{\mue}{\nlv} + 2 \var}{\frac{\mue}{\nlv}} \leq  = \frac{\frac{\var}{c} + 2 \var}{\frac{\var}{c}} =2 c + 1$$
Next, we will show that the bound is tight. Our goal is to show that, for all $\epsilon >0$, there exists some set of parameters $\ns \geq 1, \nlv \leq c \cd \frac{\mue}{\var}, \mue, \var \geq 0$ such that 
$$\frac{err_s(\{\ns, \nlv\})}{err_l(\{\ns, \nlv\})} \geq 2 c + 1 - \epsilon$$
The $\epsilon, \delta$ version of Property 4 can be written as: for all $\epsilon > 0$, there exists a $\delta > 0$ such that 
$$\frac{\ns}{\nlv} < \delta \quad \Rightarrow \quad \frac{\frac{\mue}{\nlv} + 2 \var}{\frac{\mue}{\nlv}} - \frac{err_s\p{\{\ns, \nlv}\}}{err_l\p{\{\ns, \nlv}\}}  < \epsilon$$
Setting $\ns = 1, \nlv = c \cd \frac{\mue}{\var}$ gives that Property 4 becomes: 
$$\frac{\var}{c \cd \mue} < \delta \quad \Rightarrow \quad 2 c + 1 - \frac{err_s\p{\{\ns, \nlv}\}}{err_l\p{\{\ns, \nlv}\}}  < \epsilon$$
$$\frac{\var}{c \cd \mue} < \delta \quad \Rightarrow \quad \frac{err_s\p{\{\ns, \nlv}\}}{err_l\p{\{\ns, \nlv}\}} - \epsilon > 2 c + 1 - \epsilon$$
Because $\var, \mue$ are free parameters, we can set $\var =1, \mue > \frac{1}{c \cd \delta}$ in order to satisfy the precondition. 
\end{proof}

\subsection{Proving aggregation methods are modular}
The properties in Definition \ref{def:modular} are relatively logical, desirable properties we would expect many aggregation methods would satisfy. In this section, we will theoretically prove that there exist multiple federating methods that are modular. Specifically, Theorem \ref{thrm:unifwell} shows that uniform federation is modular and Theorem \ref{thrm:finewell} shows that fine-grained federation is modular. 

\begin{restatable}{theorem}{unifwell}
\label{thrm:unifwell}
Uniform federation is modular and therefore satisfies the $2 \cd c+1$ egalitarian fairness bound from Theorem \ref{thrm:egalbound}. 
\end{restatable}

\begin{restatable}{theorem}{finewell}
\label{thrm:finewell}
Fine-grained federation is modular and therefore satisfies the $2 \cd c+1$ egalitarian fairness bound from Theorem \ref{thrm:egalbound}. 
\end{restatable}

The proofs of these theorems are relatively long and involved, so they are deferred to \ifarxiv Appendix \ref{app:proofs} \else the supplementary material (or the full version, available at \url{https://arxiv.org/abs/2112.00818}) \fi. However, we include a proof sketch for a portion of Theorem \ref{thrm:unifwell} to give intuition for the structure of these proofs. 

\begin{proof}[Proof sketch of Theorem \ref{thrm:unifwell}]
In this proof, we will denote the uniform federation error by $err^u$. Because Definition \ref{def:modular} has five components, this proof will have five sections, but in this proof sketch, we will only present Property 2. \\
\textbf{Property 2: }\\
For the second property, we wish to show that the worst case ratio of errors occurs in the two-player case. That is, 
 $$\frac{err_s^u\p{C}}{err_l^u\p{C}} \leq \frac{err_s^u\p{\{\ns, \nlv\}}}{err_l^u\p{\{\ns, \nlv}\}} $$
In order to prove this, we'll show something stronger.  Take any player $k \in \col$ (with $k \ne s, l$). Then, we will show that the derivative of the ratio with respect to the size of player $k$ ($\ndraw_k$) is always negative: 
$$\frac{\partial}{\partial\ndraw_k}\frac{err_s^u\p{C}}{err_l^u\p{C}} < 0$$
For uniform federation, we know that in the limit of a player contributing very few samples, its effect on the shared model approaches that of if it didn't participate. Put another way, 
$$\lim_{\ndraw_k \rightarrow 0} err_{s}^u(C) = err_s^u(C \setminus \{\ndraw_k\})$$
and similarly for the large player. This tells us that we can convert any coalition $\col$ into the two-player coalition $\{\ns, \nlv\}$ by sending the size of every other player to 0. If we can show that $\frac{\partial}{\partial\ndraw_k}\frac{err_s^u\p{C}}{err_l^u\p{C}} < 0$, that means that this process will only increase the ratio of the errors $\frac{err_s^u\p{C}}{err_l^u\p{C}}  $. This would show that the worst-case ratio of errors occurs in the two-player coalition. 

Next, we will start the proof. The derivative of the error ratio is given by:  $$\frac{\partial}{\partial\ndraw_k}\frac{err_s^u\p{C}}{err_l^u\p{C}}  = \frac{err_s^u(C)' \cd err_l^u(C) - err_s^u(C) \cd err_l^u(C)'}{(err_l^u(C))^2} $$
This is negative whenever: 
$$err_s^u\p{C}' \cd err_l^u\p{C}< err_s^u\p{C} \cd err_l^u\p{C}'$$
$$\frac{err_s^u\p{C}'}{err_s^u\p{C}} < \frac{err_l^u\p{C}'}{err_l^u\p{C}}$$
To show this result, we will calculate the lefthand side (ratio relating to the small player's error) and then show that it is less than the equivalent ratio for the large player. 

In order to analyze this ratio, we will need the functional form of expected error each player receives while using uniform federation, which we take from Lemma \ref{lem:error} (\cite{donahue2020model}). 
The error of the small player can be written: 
 $$\frac{\mue}{\total} + \var \frac{\sum_{i \in \col} \ndraw_i^2 + \total^2 - 2 \total \cd \ns}{\total^2}$$
 where for notational convenience we have used $\total = \sum_{i  \in \col} \ndraw_i$ to denote the total number of samples among all players in coalition $\col$. 
The derivative of the error with respect to $\ndraw_k$ is: 
 $$-\frac{\mue}{\total^2} + 2 \cd \var \frac{(\ndraw_k + \total - \ns) \cd \total - (\sum_{i \in \col} \ndraw_i^2 + \total^2 - 2 \total \cd \ns)}{\total^3}$$
 $$= \frac{-\mue \cd \total + 2\cd \var \cd \p{(\ndraw_k + \total - \ns) \cd \total - (\sum_{i \in \col} \ndraw_i^2 + \total^2 - 2 \total \cd \ns)}}{\total^3}$$
 The term we're interested in is the ratio of the derivative of the small player's error (which we just calculated) to the error of the small player. Note that the error itself can be written as: 
 $$\frac{\mue \cd \total + \var \p{\sum_{i \in \col} \ndraw_i^2 + \total^2 - 2 \total \cd \ns}}{\total^2}$$
 So the ratio of the derivative to the overall error is:  
 $$\frac{1}{\total} \cd \frac{-\mue \cd \total + 2\cd \var \cd \p{(\ndraw_k + \total - \ns) \cd \total - (\sum_{i \in \col} \ndraw_i^2 + \total^2 - 2 \total \cd \ns)}}{\mue \cd \total + \var \p{\sum_{i \in \col} \ndraw_i^2 + \total^2 - 2 \total \cd \ns}}$$
 What we would like to show is that the above term is less than the analogous term for the $\nlv$ variant, which can be symmetrically written as: 
 $$\frac{1}{\total} \cd \frac{-\mue \cd \total + 2\cd \var \cd \p{(\ndraw_k + \total - \nlv) \cd \total - (\sum_{i \in \col} \ndraw_i^2 + \total^2 - 2 \total \cd \nlv)}}{\mue \cd \total + \var \p{\sum_{i \in \col} \ndraw_i^2 + \total^2 - 2 \total \cd \nlv}}$$
 This is equivalent to proving: 
 $$\frac{B + 2 \var \cd \total \cd \ns}{A-2 \cd \var \cd \total \cd \ns} <\frac{B + 2 \var \cd \total \cd \nlv}{A-2 \cd \var \cd \total \cd \nlv} $$
 for 
 $$A = \mue \cd \total + \var \cd (\sum_{i \in \col} \ndraw_i^2 + \total^2)$$
 $$B = - \mue \cd \total + 2 \var \cd ((\ndraw_k + \total) \cd \total - (\sum_{i \in \col} \ndraw_i^2 + \total^2))$$
 We can cross multiply the inequality to get: 
 $$(B + 2 \var \cd \total \cd \ns) \cd(A-2 \cd \var \cd \total \cd \nlv) < (B + 2 \var \cd \total \cd \nlv) \cd (A-2 \cd \var \cd \total \cd \ns) $$
  $$0 < (B + 2 \var \cd \total \cd \nlv) \cd (A-2 \cd \var \cd \total \cd \ns)-(B + 2 \var \cd \total \cd \ns) \cd(A-2 \cd \var \cd \total \cd \nlv) $$
  Which simplifies to: 
  $$0 < 2 \cd (A+B)\cd (\nlv - \ns) \cd \total \cd \var$$
  Because we have assumed that $\nlv > \ns$ (the larger player has more samples than the smaller player), this is true if $A + B > 0$. We can evaluate $A+B$ as being: 
  \begin{align*}
   & A + B\\
     =& \mue \total + \var \p{\sum_{i \in \col} \ndraw_i^2 + \total^2}- \mue \total + 2 \var \p{(\ndraw_k + \total) \cd \total - \p{\sum_{i \in \col} \ndraw_i^2 + \total^2}}\\
    =& 2 \var \cd (\ndraw_k + \total) \cd \total - \var \cd \p{\sum_{i \in \col} \ndraw_i^2 - \total^2}\\
    =& 2 \var \cd \ndraw_k \cd \total + \var \cd \p{\total^2 - \sum_{i \in \col}\ndraw_i^2}\\
    >& 0
  \end{align*}
  as desired. This proves that 
$$\frac{\partial}{\partial\ndraw_k}\frac{err_s^u\p{C}}{err_l^u\p{C}} < 0$$
which by our prior reasoning indicates that for uniform federation, the worst-case error ratio occurs in the two-player coalition.  $$\frac{err_s^u\p{C}}{err_l^u\p{C}} \leq \frac{err_s^u\p{\{\ns, \nlv\}}}{err_l^u\p{\{\ns, \nlv}\}} $$
\end{proof}
For a full proof of Theorem \ref{thrm:unifwell} (including Properties 1, 3-5) and a proof of Theorem \ref{thrm:finewell} see \ifarxiv Appendix \ref{app:proofs} \else the supplementary material \fi. 

Theorems \ref{thrm:egalbound}, \ref{thrm:unifwell}, \ref{thrm:finewell}, taken together, show that while small players must get higher error than larger players in the same federation, the error that small players experience is upper bounded. Specifically, this upper bound is a multiple of the error that the larger player experiences, and depends on the size of the larger player, as a function of the noise-bias ratio. Because this bound is tight, we know that no lower bound is possible.

\section{Proportional fairness}\label{sec:prop}

\subsection{Motivation}

Next, we consider \emph{proportional fairness}, which revolves around the idea that players who contribute more data should be rewarded with lower error rates. This notion relies on the view that players do have control over how much data they obtain.  For example, imagine again two hospitals, this time with comparable levels of resources, but one of them has devoted a substantially larger share of its annual budget towards collecting high-quality local data. As a result, it has substantially more local data it can use for training. A model-sharing coalition that gives both hospitals roughly similar error rates might be seen as failing to recognize and reward that hospital's extraordinary efforts towards data collection.

Previously, Section \ref{sec:egal} we obtained results about egalitarian fairness that hold for a wide class of model-sharing. methods. In this section, on the other hand, we will show that the results for proportional fairness depend on the model-sharing method. Specifically, we will focus on whether the $\alpha$ value from Definition \ref{def:prop} is smaller than $1$, equal to 1, or larger than 1, which implies that the small player's error is lower, equal to, or great than what proportionality would suggest. Knowing this is useful because it indicates whether smaller players benefit from federation or are hurt, as compared with the goal of proportional error.  

\subsection{Fine-grained (personalized) always satisfies $\alpha \leq 1$}

Our first result is about fine-grained federation - when the aggregated model is personalized for each participating player. In this case, Theorem \ref{thrm:finestab} gives us an extremely clean result: proportional error \emph{always} satisfies $\alpha \leq 1$. That is, for any two players that are federating, $\ns \leq \nlv$, the small player always obtains error that is upper bounded by the error of the large player, multiplied by $\frac{\nlv}{\ns}$. 

\begin{restatable}{theorem}{finestab}
\label{thrm:finestab}
Optimal fine-grained federation always has $\alpha \leq 1$: the smaller player experiences \emph{lower} error than proportionality would suggest.
\end{restatable}
\begin{proof}[Proof sketch]
In order to show that $\alpha \leq 1$, we will show that: 
$$err_s^f(\col) \leq err_l^f(\col) \cd \frac{\ndraw_l}{\ndraw_s} \quad \ns \leq \nlv$$
or: 
$$err_s(\col) \cd \ndraw_s \leq err_l(\col) \cd \ndraw_l$$
where $err_s^f(\col)$ and $err_l^f(\col)$ are the error that the small and large players respectively obtain while participating in coalition $\col$ and using fine-grained federation. 

In order to analyze this, we use the functional form of error for fine-grained federation from Lemma \ref{lem:errorfine}: 
\begin{equation}\label{eq:errfine}
\mue\sum_{i=1}^{\nplayer}\vm_{ji}^2\cd \frac{1}{\ndraw_i} + \p{\sum_{i\ne j}\vm_{ji}^2 + \p{\sum_{i\ne j}\vm_{ji}}^2}\cd \var
\end{equation}
where the parameters $\vm_{ji}$ are set so as to minimize error: 
$$\vm_{jj} = \frac{1 + \var \sum_{i\ne j}\frac{1}{V_i}}{1 + V_j\sum_{i\ne j}\frac{1}{V_i}}\quad \vm_{jk} = \frac{1}{V_k}\cd \frac{V_j-\var}{1 + V_j \sum_{i\ne j}\frac{1}{V_i}} \quad k\ne j$$
for $V_i = \var + \frac{\mue}{\ndraw_i}$. 
We're interested in the term given by by $err_j^f(\col) \cd \ndraw_j$. Plugging in the $\vm_{jk}$ parameters into Eq. \eqref{eq:errfine} gives: 
\begin{align*}
 & \mue \cd \frac{(1 + \var \sum_{i\ne j}\frac{1}{V_i})^2}{(1 + V_j\sum_{i\ne j}\frac{1}{V_i})^2} \\
+ & \ndraw_j \cd \var \cd \p{
\frac{1 + V_j \sum_{i \ne j} \frac{1}{V_i}- 1 - \var \sum_{i \ne j} \frac{1}{V_i}}{1 + V_j \sum_{i \ne j} \frac{1}{V_i}}}^2\\
+& \ndraw_j \cd \sum_{i \ne j}\p{\frac{1}{V_i}\cd \frac{\frac{\mue}{\ndraw_j}}{1 + V_j \sum_{k \ne j}\frac{1}{V_k}}}^2\cd V_i
\end{align*}
Strategic simpliciation (details in full proof) shows that the above equation can be dramatically simplified to this form: 
$$\frac{\mue}{V_j \cd  \sum_{i=1}^{\nplayer} \frac{1}{V_i}} \cd  \p{1 + \var\cd  \sum_{i\ne j}^{\nplayer} \frac{1}{V_i}}$$
Next, in order to show that $\alpha \leq 1$, we need to show that: 
$$\frac{\mue}{V_s \cd  \sum_{i=1}^{\nplayer} \frac{1}{V_i}} \cd  \p{1 + \var\cd  \sum_{i\ne s}^{\nplayer} \frac{1}{V_i}} \leq \frac{\mue}{V_l \cd  \sum_{i=1}^{\nplayer} \frac{1}{V_i}} \cd  \p{1 + \var\cd  \sum_{i\ne l}^{\nplayer} \frac{1}{V_i}} $$
Further analysis shows that this holds exactly whenever: 
$$V_l \leq V_s$$
Note that because $V_i = \var + \frac{\mue}{\ndraw_i}$: 
$$\ns \leq \nlv \quad \Rightarrow \quad V_s \geq V_l$$
This means that the overall proof is concluded: using optimal fine-grained federation, sub-proportional error is always satisfied. 
\end{proof}

\subsection{Vanilla federation satisfies $\alpha \leq 1$ if individually rational}

While Theorem \ref{thrm:finestab} gives a very clean result for fine-grained federation, the situation for vanilla federation is much less clear. Lemma \ref{lem:grandcanviolateprop} shows empirically that there exist situations where a federating coalition using vanilla federation could have $\alpha$ values on either side of $\alpha =1$. 

\begin{lemma}\label{lem:grandcanviolateprop}
There exist cases where a coalition using vanilla federation satisfies $\alpha < 1$, $\alpha = 1$, or $\alpha >1$: that is, the smaller player could experience error that is lower than, equal to, or greater than proportionality would suggest.
\end{lemma}


\begin{table}[]
\centering
\begin{tabular}{|c|c|c|c|c|c|c|}
\hline
            & $err_s(\col)$ & $err_l(\col)$ & $\frac{err_s(\col)}{err_l(\col)}$ & $2 c + 1$   & $\frac{\nlv}{\ns}$ & $ \alpha$ \\  \hline
$\nlv = 20$ & 1.57                   & 0.491                  & 3.19                                                           & 5 & 3.33    &    0.96    \\ \hline
$\nlv = 30$ & 1.67                   & 0.333                   & 5                                                         & 7           & 5   & 1            \\ \hline
$\nlv = 40$ & 1.73                   & 0.251                   & 6.89                                                     & 9            & 6.67    &   1.03   \\ \hline
\end{tabular}
\caption{
The actual error rates experienced by a pair of federating agents, where $\ns = 6$ samples and $\nlv$ varies. The errors are derived using vanilla federation as given in Lemma \ref{lem:error} with parameters $\mue = 10, \var = 1$. The first two columns display the error rates experienced by both players, while the third displays the ratio of their errors. The fourth column gives a egalitarian fairness bound derived in Theorem \ref{thrm:egalbound} and the fifth column displays the ratio of the number of samples provided by each player. Finally, the sixth column displays the $\alpha$ parameter as defined in Definition \ref{def:prop} measuring how far from proportional the error rates are. 
}
\label{tab:realmotivate2}
\end{table}

\begin{proof}
This proof is constructive: we give an explicit example of such situations. 
Throughout, we use the error form given in Lemma \ref{lem:error} with $\mue = 10, \var =1$. We will consider the case where there are exactly two players: $\ns = 6$, and $\nlv$ with a varying number of samples. 

 The first row gives the error that the small and large players experience with $\nlv = 20$. Note that the implied $\alpha$ value (sixth column) is less than 1, which is because the ratio of their errors (given in the third column) is strictly less than the (inverse) ratio of their sizes (given in the fifth column). This means that the small player gets lower error than proportional scaling with respect to their sizes would suggest. The second row gives the error that the players experience with $\nlv = 30$. Here, the ratio of their errors exactly matches the ratio of their sizes: this is a case of exact proportionality. Finally, the third row gives the error the players experience when $\nlv = 40$. Here the ratio of their errors is strictly greater than the ratio of their sizes: this means $\alpha >1$. 
 
 In later analyses, we will show that all arrangements with $\alpha >1$ fail to be individually rational: at least one federating agent would prefer local learning. For the example with $\nlv =40$, the error form from Lemma \ref{lem:error} gives that the large player would obtain error of $0.25$ with local learning. This is strictly less than its federating error of $0.251$, so it is incentivized to leave the federating coalition. 

\end{proof}

Note that Table \ref{tab:realmotivate2} also illustrates the bound derived in Theorem \ref{thrm:egalbound}. Given that the error/bias ratio here is $\frac{\mue}{\var} = \frac{10}{1} = 10$, the $2 c + 1$ bound scales with the size of the largest player, in multiples of 10. The bound in the fifth column is always strictly greater than the ratio in the third column, illustrating the result from Theorem \ref{thrm:egalbound}. 

Next, Theorem \ref{thrm:propstab} gives conditions for when $\alpha \leq 1$ is guaranteed for vanilla federation: specifically, whenever the coalition is individually rational. A coalition $\col$ is \emph{individually rational} if each player in the coalition prefers being in $\col$ to doing local learning. Individual rationally is already a key goal of federated learning, indicating that many model-sharing scenarios might already achieve it \cite{cho2022federate}. 

\begin{restatable}{theorem}{egalstab}
\label{thrm:propstab}
Any individually rational coalition using vanilla federation always has $\alpha \leq 1$: the smaller player experiences lower error than proportionality would suggest.
\end{restatable}

As mentioned previously, $\alpha \leq 1$ means that smaller players (contributing fewer samples to the federating model) experience error that is \emph{lower} than proportionality would suggest. This may or may not be desirable, depending on the motivation for model sharing (as described in Section \ref{sec:fairdef}).  The full proof of Theorem \ref{thrm:propstab} is given in \ifarxiv Appendix \ref{app:proofs} \else the supplementary material \fi, but we include a proof sketch below. 
\begin{proof}[Proof Sketch]
We will show this result by proving the contrapositive: if $\alpha\leq1$ is violated, then at least one player wishes to leave that coalition for local learning (so individual rationality is violated). 

In our analysis, we will consider two players $\ns < \nlv$, though the players could be arbitrarily situated with respect to $ \frac{\mue}{\var}$. We will assume that all of the players are federating together in coalition $A$. For notational convenience, we will write the federating coalition $A$ of interest as $A = \col \cup \{\nlv\}$, so $\col$ refers to all players except $\nlv$. 

The case where $\alpha > 1$ as: 
$$err_{\si}(A) \cd \ns \geq err_{\li}(A) \cd \nlv$$

The full proof, in \ifarxiv Appendix \ref{app:proofs} \else the supplementary material \fi, shows that this inequality can be rewritten as: 
\begin{equation}\label{eq:subpropviol_main} \frac{-2 \var \cd \total_C + \frac{\mue}{\ns} \cd \total_C + \frac{\var}{\ns} \cd \p{\sum_{i \in C} \ndraw_i^2 + \total_C^2}}{2 \var-\frac{\mue}{\ns} }  \leq \nlv
\end{equation}
Equation \eqref{eq:subpropviol_main} gives a lower bound on how many samples the large player has to have before $\alpha \leq 1$ is violated. Next, the proof obtains a similar lower bound for how many samples the large player would need before it would prefer to defect: we will show this lower bound is \emph{lower} than Equation \eqref{eq:subpropviol_main} above. 

The case when the large player wishes to defect from $A$ to local learning (or is ambivalent about defecting) is when: 
$$err_{\li}(\alone) \leq err_{\li}(A)$$
The full version of the proof shows that this can be simplified to: 
\begin{equation}\label{eq:unstable_main} \frac{-2 \var \cd \total_C + \frac{\mue}{\ns} \cd \total_C + \frac{\var}{\ns} \cd \p{\sum_{i \in C} \ndraw_i^2 + \total_C^2}}{2 \var-\frac{\mue}{\ns} }  \leq \nlv
\end{equation}

The full proof concludes by showing that the bound in Equation \eqref{eq:unstable_main} is lower than the bound in Equation \eqref{eq:subpropviol_main}: any situation where the $\nlv$ term is large enough so that $\alpha>1$, we know that the large player already wishes to leave.
\end{proof}

Overall, the bound in Theorem \ref{thrm:propstab} is different from the results on egalitarian fairness in Section \ref{sec:egal}. For egalitarian fairness, we obtained a single bound that held for a wide range of model aggregation methods. In this section, by contrast, we show that the bounds we obtain are specific to the model aggregation method. In Theorem \ref{thrm:finestab}, we show that for fine-grained federation, $\alpha \leq 1$ is always satisfied. However, for vanilla federation, we obtain a guarantee that holds whenever a coalition is individually rational: whenever all collaborating players obtain lower error than they would in local learning. The result in Theorem \ref{thrm:propstab} could be useful if participating is truly voluntary: then, any coalition where smaller players experience error greater than proportional could never exist, because it would unravel into local learning. However, this result could still be useful even if joining a coalition isn't something individual players could opt out of. Theorem \ref{thrm:propstab} shows that any coalition with $\alpha >1$ is one that also has the problem that some players receive higher error than they would in local learning - which indicates a broader failure of model-sharing, which may prompt any organizer to change their model aggregation methods or coalition structure. 

\section{Discussion}

In this work, we have used a theoretical model to analyze the fairness properties of model-sharing games. We formalized two notions of fairness already under analysis in data cooperative and federated learning research and gave conditions for when those fairness properties are or are not satisfied. For egalitarian fairness, we gave a tight multiplicative bound on the gap in error rates between any two federating players. For proportional fairness, we showed that sub-proportional fairness ($\alpha \leq 1$) is satisfied always for fine-grained federation, and for any individually rational coalition for vanilla federation. 

Future research could consider questions in two directions. First: which other notions of fairness might be relevant? It might be the case that an additive version of egalitarian fairness would be more appropriate for certain situations, for example, or a version of proportional fairness based on some metric related to the quality of data points, rather than how their quantity. Secondly, which fairness properties might other methods of model aggregation display? While we have analyzed two federation methods (vanilla and fine-grained), it could be fruitful to analyze other methods to determine how well they achieve certain fairness results. 

%
\begin{acks}
This work was supported in part by a Simons Investigator Award, a Vannevar Bush Faculty Fellowship, MURI grant W911NF-19-0217, AFOSR grant FA9550-19-1-0183, ARO grant W911NF19-1-0057, a Simons Collaboration grant, a grant from the MacArthur Foundation, and NSF grant DGE-1650441. We are grateful to the AI, Policy, and Practice working group at Cornell for invaluable discussions. 
\end{acks}

\bibliographystyle{ACM-Reference-Format}
\bibliography{sample-base}

\ifarxiv 
\appendix
\clearpage 

\section{Proofs}\label{app:proofs}

\unifwell*

\begin{proof}
In this proof, we will denote the uniform federation error by $err^u$. Because Definition \ref{def:modular} has five components, this proof will have five sections. \\
\textbf{Property 1: }\\
For the first property, we wish to show that (for any federating coalition $\col$), the large player always has strictly lower error than the small player, or: 
$$\frac{err_s(\col)}{err_l(\col)} \leq 1 \quad \Leftrightarrow \quad err_s(\col) \geq err_l(\col)$$
Using the form of error found in Lemma \ref{lem:error} (Equation \eqref{eq:err}), we can rewrite this as: 
$$\frac{\mue}{\sum_{i \in \col} \ndraw_i} + \var \cd \frac{\sum_{i \in \col}\ndraw_i^2  - \ns^2+ \p{\sum_{i\in \col}\ndraw_i - \ns}^2}{\p{\sum_{i \in \col}\ndraw_i}^2}$$
$$\geq \frac{\mue}{\sum_{i \in \col} \ndraw_i} + \var \cd \frac{\sum_{i \in \col}\ndraw_i^2  - \nlv^2+ \p{\sum_{i\in \col}\ndraw_i - \nlv}^2}{\p{\sum_{i \in \col}\ndraw_i}^2}$$
Cancelling common terms: 
\begin{align*}
\sum_{i \in \col}\ndraw_i^2  - \ns^2+ \p{\sum_{i\in \col}\ndraw_i - \ns}^2 &\geq \sum_{i \in \col}\ndraw_i^2  - \nlv^2+ \p{\sum_{i\in \col}\ndraw_i - \nlv}^2
\end{align*}
Which expands out to: 
\begin{align*}
 - \ns^2+ \p{\sum_{i\in \col}\ndraw_i}^2 -2 \cd \ns \cd \p{\sum_{i\in \col}\ndraw_i}  + \ns^2\\
 \geq  - \nlv^2+ \p{\sum_{i\in \col}\ndraw_i}^2 -2 \cd \nlv \cd \p{\sum_{i\in \col}\ndraw_i} + \nlv^2
 \end{align*}
 Cancelling common terms gives: 
\begin{align*}
 -2 \cd \ns \cd \p{\sum_{i\in \col}\ndraw_i}  &\geq  -2 \cd \nlv \cd \p{\sum_{i\in \col}\ndraw_i}\\
  \ns  &\leq  \nlv
\end{align*}
This means that the small player has higher error whenever it has fewer samples than the large player- and strictly higher error whenever it has strictly fewer samples. \\
\textbf{Property 2: }\\
For the second property, we wish to show that the worst case ratio of errors occurs in the two-player case. That is, 
 $$\frac{err_s^u\p{C}}{err_l^u\p{C}} \leq \frac{err_s^u\p{\{\ns, \nlv\}}}{err_l^u\p{\{\ns, \nlv}\}} $$
In order to prove this, we'll show something stronger.  Take any player $k$ (with $k \ne s, l$). Then, we will show that that the derivative of the ratio with respect to the size of player $k$ ($\ndraw_k$) is always negative: 
$$\frac{\partial}{\partial\ndraw_k}\frac{err_s^u\p{C}}{err_l^u\p{C}} < 0$$
For uniform federation, we know that 
$$\lim_{\ndraw_k \rightarrow 0} err_{s}^u(C) = err_s^u(C \setminus \{\ndraw_k\})$$
and similarly for the large player. Then, we can convert $C$ into $\{\ns, \nlv\}$ by sending the size of every other player to 0: by the result we are trying to prove, this will only ever increase the ratio of their errors. 

Next, we will start the proof. $$\frac{\partial}{\partial\ndraw_k}\frac{err_s^u\p{C}}{err_l^u\p{C}}  = \frac{err_s^u(C)' \cd err_l^u(C) - err_s^u(C) \cd err_l^u(C)'}{(err_l^u(C))^2} $$
This is negative whenever: 
$$err_s^u\p{C}' \cd err_l^u\p{C}< err_s^u\p{C} \cd err_l^u\p{C}'$$
$$\frac{err_s^u\p{C}'}{err_s^u\p{C}} < \frac{err_l^u\p{C}'}{err_l^u\p{C}}$$
To show this result, we will calculate the lefthand side (ratio relating to the small player's error) and then show that it is less than the equivalent ratio for the large player. 
The error of the small player can be written: 
 $$\frac{\mue}{\total} + \var \frac{\sum_{i \in \col} \ndraw_i^2 + \total^2 - 2 \total \cd \ns}{\total^2}$$
The derivative with respect to $\ndraw_k$ is: 
 $$-\frac{\mue}{\total^2} + 2 \cd \var \frac{(\ndraw_k + \total - \ns) \cd \total - (\sum_{i \in \col} \ndraw_i^2 + \total^2 - 2 \total \cd \ns)}{\total^3}$$
 $$= \frac{-\mue \cd \total + 2\cd \var \cd \p{(\ndraw_k + \total - \ns) \cd \total - (\sum_{i \in \col} \ndraw_i^2 + \total^2 - 2 \total \cd \ns)}}{\total^3}$$
 The term we're interested in is the ratio of the derivative of the small player's error (which we just calculated) to the error of the small player. Note that the error itself can be written as: 
 $$\frac{\mue \cd \total + \var \p{\sum_{i \in \col} \ndraw_i^2 + \total^2 - 2 \total \cd \ns}}{\total^2}$$
 So the ratio of the derivative to the overall error is:  
 $$\frac{1}{\total} \cd \frac{-\mue \cd \total + 2\cd \var \cd \p{(\ndraw_k + \total - \ns) \cd \total - (\sum_{i \in \col} \ndraw_i^2 + \total^2 - 2 \total \cd \ns)}}{\mue \cd \total + \var \p{\sum_{i \in \col} \ndraw_i^2 + \total^2 - 2 \total \cd \ns}}$$
 What we would like to show is that the above term is less than the analogous term for the $\nlv$ variant, which can be symmetrically written as: 
 $$\frac{1}{\total} \cd \frac{-\mue \cd \total + 2\cd \var \cd \p{(\ndraw_k + \total - \nlv) \cd \total - (\sum_{i \in \col} \ndraw_i^2 + \total^2 - 2 \total \cd \nlv)}}{\mue \cd \total + \var \p{\sum_{i \in \col} \ndraw_i^2 + \total^2 - 2 \total \cd \nlv}}$$
 This is equivalent to proving: 
 $$\frac{B + 2 \var \cd \total \cd \ns}{A-2 \cd \var \cd \total \cd \ns} <\frac{B + 2 \var \cd \total \cd \nlv}{A-2 \cd \var \cd \total \cd \nlv} $$
 for 
 $$A = \mue \cd \total + \var \cd (\sum_{i \in \col} \ndraw_i^2 + \total^2) \quad B = - \mue \cd \total + 2 \var \cd ((\ndraw_k + \total) \cd \total - (\sum_{i \in \col} \ndraw_i^2 + \total^2))$$
 We can cross multiply the inequality to get: 
 $$(B + 2 \var \cd \total \cd \ns) \cd(A-2 \cd \var \cd \total \cd \nlv) < (B + 2 \var \cd \total \cd \nlv) \cd (A-2 \cd \var \cd \total \cd \ns) $$
  $$0 < (B + 2 \var \cd \total \cd \nlv) \cd (A-2 \cd \var \cd \total \cd \ns)-(B + 2 \var \cd \total \cd \ns) \cd(A-2 \cd \var \cd \total \cd \nlv) $$
  Which simplifies to: 
  $$0 < 2 \cd (A+B)\cd (\nlv - \ns) \cd \total \cd \var$$
  Because we have assumed that $\nlv > \ns$, this is true if $A + B > 0$. We can evaluate $A+B$ as being: 
  \begin{align*}
     =& \mue \cd \total + \var \cd \p{\sum_{i \in \col} \ndraw_i^2 + \total^2} - \mue \cd \total \\
      & + 2 \var \cd ((\ndraw_k + \total) \cd \total - (\sum_{i \in \col} \ndraw_i^2 + \total^2))\\
    =& 2 \var \cd (\ndraw_k + \total) \cd \total) - \var \cd \p{\sum_{i \in \col} \ndraw_i^2 - \total^2}\\
    =& 2 \var \cd \ndraw_k \cd \total + 2 \var \cd \total^2 - \var \cd \p{\sum_{i \in \col} \ndraw_i^2 - \total^2}\\
    =& 2 \var \cd \ndraw_k \cd \total + \var \cd \p{\total^2 - \sum_{i \in \col}\ndraw_i^2}\\
    >& 0
  \end{align*}
  as desired. \newline \\
  The remaining properties all relate to the ratio of errors for the two-player group $\{\ns, \nlv\}$. For uniform federation, this ratio can be written as: 
  $$\frac{err_s^u\p{\{\ns, \nlv\}}}{err_l^u\p{\{\ns, \nlv}\}} = \frac{\mue \cd (\ns + \nlv) + 2 \var \cd \nlv^2}{\mue \cd (\ns + \nlv) + 2 \var \cd \ns^2}$$
  \textbf{Property 3: }\\

  This property relates to the derivative of the ratio with respect to the size of the large player ($\nlv$): 
  $$\frac{\partial}{\partial\nlv} \frac{\mue \cd (\ns + \nlv) + 2 \var \cd \nlv^2}{\mue \cd (\ns + \nlv) + 2 \var \cd \ns^2} $$
  This derivative is given by: 
  $$\frac{(\mue  + 4 \cd \var \cd \nlv) (\mue (\ns + \nlv) +2 \var \cd \ns^2) - (\mue (\ns + \nlv) + 2 \var \cd \nlv^2) \mue}{(\mue \cd (\ns + \nlv) + 2 \var \cd \ns^2)^2}$$
  The derivative is \emph{positive} whenever: 
  $$ (\mue  + 4 \cd \var \cd \nlv) \cd (\mue \cd (\ns + \nlv) +2 \var \cd \ns^2) - (\mue \cd (\ns + \nlv) + 2 \var \cd \nlv^2) \cd \mue > 0 $$
  Pulling over terms and expanding gives: 
  $$\mue ^2 \cd (\ns + \nlv) + 2 \cd \mue \cd \var \cd \ns^2 + 4 \cd \var \cd \nlv \cd \mue \cd (\ns + \nlv) + 8 \cd \sigma^4 \cd \nlv \cd \ns^2 >$$
  $$ \mue^2 \cd (\ns + \nlv) + 2 \var \cd \mue \cd \nlv^2 $$
  Which simplifies to: 
  $$ 2 \cd \mue \cd \var \cd \ns^2 + 2 \cd \var \cd \nlv \cd \mue \cd (2\ns + \nlv) + 8 \cd \sigma^4 \cd \nlv \cd \ns^2 >0$$
  \textbf{Property 4: }\\
  This property relates to the derivative of the ratio with respect to the size of the small player ($\ns$): 
  $$\frac{\partial}{\partial\ns} \frac{\mue \cd (\ns + \nlv) + 2 \var \cd \nlv^2}{\mue \cd (\ns + \nlv) + 2 \var \cd \ns^2} $$
  This derivative is given by: 
  $$\frac{\mue (\mue (\ns + \nlv) + 2 \var \cd \ns^2) - (\mue (\ns + \nlv) + 2 \var \cd \nlv^2) \cd (\mue + 4 \cd \var \cd \ns)}{(\mue (\ns + \nlv) + 2 \var \cd \ns^2)^2}$$
  The derivative is \emph{negative} whenever: 
  $$\mue \cd (\mue \cd (\ns + \nlv) + 2 \var \cd \ns^2) < (\mue \cd (\ns + \nlv) + 2 \var \cd \nlv^2) \cd (\mue + 4 \cd \var \cd \ns)$$
  Expanding gives: 
  \begin{align*}
\mue \cd (\mue \cd (\ns + \nlv) + 2 \var \cd \ns^2)\\
 < \mue^2 \cd (\ns + \nlv) + 2 \var \cd \nlv^2 \cd \mue + 4 \cd \var \cd \ns \cd \mue  \cd (\ns + \nlv) + 2 \var \cd \mue  \cd \nlv^2
  \end{align*}
  Simplifying: 
  $$0 <  2 \var \cd \nlv^2 \cd \mue + 2 \cd \var \cd \ns \cd \mue  \cd (\ns + 2\nlv) + 2 \var \cd \mue  \cd \nlv^2 $$
  which is satisfied.\\
  \textbf{Property 5:}\\
    This can be found by rewriting and applying the limit. 
    \begin{align*}
    \lim_{\frac{\ns}{\nlv} \rightarrow 0}\frac{err_s\p{\{\ns, \nlv}\}}{err_l\p{\{\ns, \nlv}\}} =& \lim_{\frac{\ns}{\nlv} \rightarrow 0} \frac{\mue \cd (\ns + \nlv) + 2 \var \cd \nlv^2}{\mue \cd (\ns + \nlv) + 2 \var \cd \ns^2}\\
    =& \lim_{\frac{\ns}{\nlv} \rightarrow 0} \frac{\mue \cd \frac{\ns}{\nlv} + \mue + 2 \var \cd \nlv}{\mue \cd \frac{\ns}{\nlv} + \mue + 2 \var \cd \ns\cd \frac{\ns}{\nlv}}   \\
    =&\frac{\mue + 2 \var \cd \nlv}{\mue} \\
    =& \frac{\frac{\mue}{\nlv} + 2 \var}{\frac{\mue}{\nlv}}
    \end{align*}
  
\end{proof}

\begin{lemma}\label{lem:finegrainederror}
Consider a set of $\nplayer$ federating players, using optimal fine-grained federation. The error player $j$ (with $\ndraw_j$ samples) experiences can be given by: 
$$\frac{\frac{\mue}{\ndraw_j}}{V_j \cd T} \cd  \p{1 + \var \p{T - \frac{1}{V_j}}}$$
with $V_i = \var + \frac{\mue}{\ndraw_i}$ and $T = \sum_{i=1}^{\nplayer} \frac{1}{V_i}$. 
\end{lemma}
\begin{proof}
Lemma 4.2 in \cite{donahue2020model} gives the expected error as:
$$\mue\sum_{i=1}^{\nplayer}\vm_{ji}^2\cd \frac{1}{\ndraw_i} + \p{\sum_{i\ne j}\vm_{ji}^2 + \p{\sum_{i\ne j}\vm_{ji}}^2}\cd \var$$
Fine-grained federation requires that $\vm_{jj} + \sum_{i\ne j}\vm_{ji} = 1$, so $\vm_{jj} = 1 - \sum_{i \ne j} \vm_{ji}$, or $\sum_{i\ne j} \vm_{ji} = 1- \vm_{jj}$. We can then write the error out as: 
$$\mue \sum_{i=1}^{\nplayer}\frac{\vm_{ji}^2}{\ndraw_i} + \p{\sum_{i\ne j}\vm_{ji}^2 + \p{1 - \vm_{jj}}^2} \cd \var$$
Or: 
$$\mue \cd \frac{\vm_{jj}^2}{\ndraw_j} + (1-\vm_{jj})^2 \cd \var + \sum_{i \ne j}\vm_{ji}^2 \cd \p{\frac{\mue}{\ndraw_i} + \var}$$
We note that this last form has three components (the first involving $\mue$ and the latter two involving $\var$). 

Lemma 7.1 from the same work shows that the optimal $\vm_{ji}$ weights can be calculated as follows: Define $V_i = \var + \frac{\mue}{\ndraw_i}$. Then, the value of $\{\vm_{ji}\}$ that minimizes player $j$'s error is: 
$$\vm_{jj} = \frac{1 + \var \sum_{i\ne j}\frac{1}{V_i}}{1 + V_j\sum_{i\ne j}\frac{1}{V_i}}$$
$$\vm_{jk} = \frac{1}{V_k}\cd \frac{V_j-\var}{1 + V_j \sum_{i\ne j}\frac{1}{V_i}} \quad k\ne j$$
This analysis follows by plugging in the optimal weights and then simplifying. 
The first component of the error becomes: 
$$\mue \cd \frac{v_{jj}^2}{\ndraw_j} = \frac{\mue}{\ndraw_j} \cd \frac{(1 + \var \sum_{i\ne j}\frac{1}{V_i})^2}{(1 + V_j\sum_{i\ne j}\frac{1}{V_i})^2}$$
The next component becomes: 
$$ (1-\vm_{jj})^2 \cd \var =  \var \cd \p{
\frac{1 + V_j \sum_{i \ne j} \frac{1}{V_i}- 1 - \var \sum_{i \ne j} \frac{1}{V_i}}{1 + V_j \sum_{i \ne j} \frac{1}{V_i}}}^2 $$
$$= \var \cd \frac{\frac{\mue^2}{\ndraw_j^2} \cd (\sum_{i \ne j}\frac{1}{V_i})^2}{(1 + V_j \sum_{i \ne j} \frac{1}{V_i})^2}$$
The last component: 
$$\sum_{i \ne j}\vm_{ji}^2 \cd \p{\frac{\mue}{\ndraw_i} + \var} =  \sum_{i \ne j}\p{\frac{1}{V_i}\cd \frac{\frac{\mue}{\ndraw_j}}{1 + V_j \sum_{k \ne j}\frac{1}{V_k}}}^2\cd V_i $$
$$=\frac{\mue^2}{\ndraw_j^2} \cd \sum_{i \ne j}\frac{1}{V_i}\cd \frac{1}{(1 + V_j \sum_{k \ne j}\frac{1}{V_k})^2} = \frac{\mue^2}{\ndraw_j^2} \cd \frac{\sum_{i \ne j} \frac{1}{V_i}}{(1 + V_j \cd \sum_{k\ne j} \frac{1}{V_k})^2}$$
When we combine them together, we can pull out a common factor of: 
$$\frac{\frac{\mue}{\ndraw_j}}{(1 + V_j \cd \sum_{i \ne j} \frac{1}{V_i})^2}$$
The other terms become: 
$$\p{1 + \var \sum_{i \ne j} \frac{1}{V_i}}^2 + \var \cd \frac{\mue}{\ndraw_j} \cd \p{\sum_{i \ne j} \frac{1}{V_i}}^2 + \frac{\mue}{\ndraw_j} \cd \sum_{i \ne j} \frac{1}{V_i}$$ 
Combining these together gives: 
$$\frac{\frac{\mue}{\ndraw_j}}{\p{1 + V_j \cd \p{T - \frac{1}{V_j}}}^2} \cd \p{1 + \var \p{T - \frac{1}{V_j}}}^2 $$
$$+ \var \cd \frac{\mue}{\ndraw_j} \cd \p{T - \frac{1}{V_j}}^2 + \frac{\mue}{\ndraw_j} \cd \p{T - \frac{1}{V_j}}$$
for $T = \sum_{i=1}^{\nplayer} \frac{1}{V_i}$. Next, we begin simplifying, noting that: 
$$\frac{\frac{\mue}{\ndraw_j}}{\p{1 + V_j \cd \p{T - \frac{1}{V_j}}}^2} = \frac{\frac{\mue}{\ndraw_j}}{\p{1 + V_j \cd T - 1}^2} = \frac{\frac{\mue}{\ndraw_j}}{V_j^2 \cd T^2}$$
We can also simplify the other coefficient:

\begin{align*}
 &  \p{1 + \var \p{T - \frac{1}{V_j}}}^2 + \var \cd \frac{\mue}{\ndraw_j} \cd \p{T - \frac{1}{V_j}}^2 + \frac{\mue}{\ndraw_j} \cd \p{T - \frac{1}{V_j}}\\
=&\p{1 + \var \p{T - \frac{1}{V_j}}}^2 + \frac{\mue}{\ndraw_j} \cd \p{T - \frac{1}{V_j}} \cd \p{\var \cd \p{T - \frac{1}{V_j}} + 1}\\
 =& \p{1 + \var \p{T - \frac{1}{V_j}}} \cd \p{1 + \var \p{T - \frac{1}{V_j}} + \frac{\mue}{\ndraw_j} \cd \p{T - \frac{1}{V_j}}}\\
 =& \p{1 + \var \p{T - \frac{1}{V_j}}} \cd \p{1 +\p{\var +  \frac{\mue}{\ndraw_j} }\cd \p{T - \frac{1}{V_j}}}\\
 =& \p{1 + \var \p{T - \frac{1}{V_j}}} \cd \p{1 +V_j\cd \p{T - \frac{1}{V_j}}}\\
 =& \p{1 + \var \p{T - \frac{1}{V_j}}} \cd \p{1 +V_j\cd T -1}\\
 =& \p{1 + \var \p{T - \frac{1}{V_j}}} \cd V_j\cd T 
\end{align*}
Next, we combine the two terms together to get:  
$$\frac{\frac{\mue}{\ndraw_j}}{V_j \cd T} \cd  \p{1 + \var \p{T - \frac{1}{V_j}}}$$
as desired. 
\end{proof}

\finewell*
\begin{proof}
In this proof, we will denote fine federation error by $err^f$. This proof will be similar to that of Lemma \ref{lem:unifwell} and will again have five components corresponding to the five sections of Definition \ref{def:modular}.  \\
\textbf{Property 1:}\\
For the first property, we wish to show that, for any federating coalition, the small player gets higher error than the large player. That is, we wish to show: 
$$err^f_s(\col) \geq err^f_l(\col)$$
From Lemma \ref{lem:finegrainederror}, we know we can write each player's error as: 
$$  \frac{\mue}{\ns \cd V_s \cd T} \cd  \p{1 + \var \p{S+ \frac{1}{V_l}}} \geq \frac{\mue}{\nlv \cd V_l \cd T} \cd  \p{1 + \var \p{S + \frac{1}{V_s}}}$$
where we have $V_i = \frac{\mue}{\ndraw_i} + \var$, $T =\sum_{i \in \col } \frac{1}{V_i}$ and $S = T - \frac{1}{V_s} - \frac{1}{V_l}$. We know that $V_i, T >0$, so we can cancel common terms: 
\begin{align*}
\frac{1}{\ns \cd V_s} \cd  \p{1 + \var \p{S+ \frac{1}{V_l}}} &\geq \frac{1}{\nlv \cd V_l} \cd  \p{1 + \var \p{S + \frac{1}{V_s}}}\\
\nlv \cd V_l \p{1 + \var \p{S+ \frac{1}{V_l}}} &\geq \ns \cd V_s\cd  \p{1 + \var \p{S + \frac{1}{V_s}}}\\
\nlv \cd V_l + \var \cd S \cd \nlv \cd V_l + \var \cd \nlv &\geq \ns \cd V_s + \var \cd S \cd \ns \cd V_s + \var \cd \ns
\end{align*}
Next, we can plug in for $V_i = \frac{\mue}{\ndraw_i} + \var$: 
\begin{align*}
\mue + \nlv \cd \var + \var \cd S \p{\mue + \nlv \cd \var} + \var \cd \nlv \\
\geq \mue + \ns \cd \var + \var \cd S \p{\mue + \ns \cd \var} + \var \cd \ns\\
\end{align*}
Cancelling common terms gives: 
\begin{align*}
2 \cd \nlv \cd \var + \nlv \cd S \cd \sigma^4 & \geq  2 \cd \ns \cd \var + S \cd \ns \cd \sigma^4\\
\nlv &\geq \ns 
\end{align*}
This result tells us that the small player has higher error whenever it has fewer samples than the large player - and strictly higher error whenever it has strictly fewer samples. \\
\textbf{Property 2: }\\
For the second property, we again wish to show that the worst case ratio of errors occurs in the two-player case, or 
 $$\frac{err_s^f\p{C}}{err_l^f\p{C}} \leq \frac{err_s^f\p{\{\ns, \nlv\}}}{err_l^f\p{\{\ns, \nlv}\}} $$

From Lemma \ref{lem:finegrainederror}, we know that we can write the ratio of errors as: 
$$ \frac{ \frac{\mue}{\ns \cd V_s \cd T} \cd  \p{1 + \var \p{S+ \frac{1}{V_l}}}}{\frac{\mue}{\nlv \cd V_l \cd T} \cd  \p{1 + \var \p{S + \frac{1}{V_s}}}}$$
where we have $V_i = \frac{\mue}{\ndraw_i} + \var$, $T =\sum_{i \in \col } \frac{1}{V_i}$ and $S = T - \frac{1}{V_s} - \frac{1}{V_l}$. The ratio of errors can be simplified further:
$$=\frac{ \nlv \cd V_l\cd  \p{1 + \var \p{S+ \frac{1}{V_l}}}}{\ns \cd V_s \cd \p{1 + \var \p{S + \frac{1}{V_s}}}} = \frac{\nlv \cd V_l + \var \cd S \cd \nlv \cd V_l + \nlv \cd \var}{\ns \cd V_s + \var \cd S \cd \ns \cd V_s + \ns \cd \var}$$
Again, we will prove something stronger: that this ratio decreases as $S$ increases. Note that $S = \sum_{i \in \col, i \ne s, l}\frac{1}{V_i} = \sum_{i \in \col, i \ne s, l}\frac{1}{\frac{\mue}{\ndraw_i} + 2 \var}$, so as $\ndraw_i$ decreases, $S$ decreases as well.  We can rewrite the equation as: 
$$\frac{a+b \cd S}{c+d \cd S}$$
for: 
$$a=\nlv \cd (V_l + \var) \quad b = \var \cd \nlv \cd V_l$$
$$c = \ns \cd (V_s + \var) \quad d = \var \cd \ns \cd V_s$$
The key question relates to the derivative of the equation with respect to $S$: 
$$\frac{\partial}{\partial S}\frac{a+b \cd S}{c+d \cd S} = \frac{ b \cd (c + d\cd S) - (a + b \cd S) \cd d}{(c+d S)^2} = \frac{ b \cd c - a \cd d}{(c+d \cd  S)^2}$$
This is negative if: 
$$b\cd c<a \cd d$$
Plugging in for the values gives: 
$$\var \cd \nlv \cd \ns \cd V_l\cd  (V_s + \var)< \var \cd \nlv \cd \ns \cd V_s\cd  (V_l + \var)$$
$$\sigma^4 \cd \nlv \cd V_l\cd \ns < \sigma^4 \cd \ns \cd \nlv \cd  V_s$$
$$ V_l< V_s$$
$$ \frac{\mue}{\nlv} + \var < \frac{\mue}{\ns} + \var$$
which is satisfied because $\nlv > \ns$.  \newline \\
The remaining properties all relate to the ratio of errors for the two-player group $\{\ns, \nlv\}$. For fine-grained federation, this ratio can be written as: 
  $$\frac{err_s^f\p{\{\ns, \nlv\}}}{err_l^f\p{\{\ns, \nlv}\}} = \frac{\nlv \cd V_l \cd \p{1 + \var\cd \frac{1}{V_l}}}{\ns \cd V_s \cd  \p{1 + \var\cd  \frac{1}{V_s}}}$$
  $$= \frac{\nlv \cd V_l + \var\cd \nlv}{\ns \cd V_s + \var\cd  \ns}  = \frac{\nlv}{\ns} \cd \frac{2\var + \frac{\mue}{\nlv}}{2 \var + \frac{\mue}{\ns}} = \frac{2 \var \cd \nlv + \mue}{2 \var \cd \ns + \mue}$$
  The remaining three properties for fine-grained federation are extremely straightforward: \\
\textbf{Property 3: }\\
This property relates to the derivative of the ratio of the errors with respect to the size of the large player ($\nlv$): 
$$\frac{\partial}{\partial \nlv} \frac{2 \var \cd \nlv + \mue}{2 \var \cd \ns + \mue} = 2 \var > 0$$
\textbf{Property 4: }\\
This property relates to the derivative of the ratio of the errors with respect to the size of the small player ($\ns$):
$$\frac{\partial}{\partial \ns} \frac{2 \var \cd \nlv + \mue}{2 \var \cd \ns + \mue} = \frac{0 - 2 \var \cd (2 \var \cd \nlv + \mue)}{(2 \var \cd \ns + \mue)^2} < 0$$
\textbf{Property 5: }\\
Finally, this property can be seen by applying the limit: 
$$\lim_{\frac{\ns}{\nlv} \rightarrow 0}\frac{err_s\p{\{\ns, \nlv}\}}{err_l\p{\{\ns, \nlv}\}} =\lim_{\frac{\ns}{\nlv} \rightarrow 0} \frac{2 \var \cd \nlv + \mue}{2 \var \cd \ns + \mue}$$
$$= \lim_{\frac{\ns}{\nlv} \rightarrow 0} \frac{2 \var + \frac{\mue}{\nlv}}{2 \var \cd \frac{\ns}{\nlv} + \frac{\mue}{\nlv}} = \frac{ \frac{\mue}{\nlv} + 2 \var}{\frac{\mue}{\nlv}}$$
\end{proof}

\finestab*
\begin{proof}
As mentioned previously, the expected error for the fine-grained case looks like: 
$$\mue \sum_{i=1}^{\nplayer}\frac{\vm_{ji}^2}{\ndraw_i} + \p{\sum_{i\ne j}\vm_{ji}^2 + \p{\sum_{i\ne j}\vm_{ji}}^2} \cd \var$$
We require that $\vm_{jj} + \sum_{i\ne j}\vm_{ji} = 1$, so $\vm_{jj} = 1 - \sum_{i \ne j} \vm_{ji}$, or $\sum_{i\ne j} \vm_{ji} = 1- \vm_{jj}$. We can then write the error out as: 
$$\mue \sum_{i=1}^{\nplayer}\frac{\vm_{ji}^2}{\ndraw_i} + \p{\sum_{i\ne j}\vm_{ji}^2 + \p{1 - \vm_{jj}}^2} \cd \var$$
Or: 
$$\mue \cd \frac{\vm_{jj}^2}{\ndraw_j} + (1-\vm_{jj})^2 \cd \var + \sum_{i \ne j}\vm_{ji}^2 \cd \p{\frac{\mue}{\ndraw_i} + \var}$$
For optimal fine-grained federation, we have that the weights are given by: 
$$\vm_{jj} = \frac{1 + \var \sum_{i\ne j}\frac{1}{V_i}}{1 + V_j\sum_{i\ne j}\frac{1}{V_i}}$$
$$\vm_{jk} = \frac{1}{V_k}\cd \frac{V_j-\var}{1 + V_j \sum_{i\ne j}\frac{1}{V_i}} \quad k\ne j$$
where we have $V_i = \var + \frac{\mue}{\ndraw_i}$. 
In order for sub-proportional error ($\alpha \leq 1$) to be satisfied, we need: 
$$err_s(\col) \leq err_l(\col) \cd \frac{\ndraw_l}{\ndraw_s} \quad \ns \leq \nlv$$
$$err_s(\col) \cd \ndraw_s \leq err_l(\col) \cd \ndraw_l \quad \ns \leq \nlv$$

We will start by taking the (weighted) error form for fine-grained federation, plugging in for the $\vm_{jj}, \vm_{jk}$ weights, and then simplifying. \\

The first component of the error (involving $\mue$ and $v_{jj}$) becomes: 
$$\ndraw_j \cd \mue \cd \frac{v_{jj}^2}{\ndraw_j} = \mue \cd \frac{(1 + \var \sum_{i\ne j}\frac{1}{V_i})^2}{(1 + V_j\sum_{i\ne j}\frac{1}{V_i})^2}$$
The component involving $\var$ and  $v_{jj}$ becomes: 

\begin{align*}
\ndraw_j \cd (1-\vm_{jj})^2 \cd \var &= \ndraw_j \cd \var \cd \p{
\frac{1 + V_j \sum_{i \ne j} \frac{1}{V_i}- 1 - \var \sum_{i \ne j} \frac{1}{V_i}}{1 + V_j \sum_{i \ne j} \frac{1}{V_i}}}^2 \\
& = \ndraw_j \cd \var \cd \frac{\frac{\mue^2}{\ndraw_j^2} \cd (\sum_{i \ne j}\frac{1}{V_i})^2}{(1 + V_j \sum_{i \ne j} \frac{1}{V_i})^2}\\
&=\var \cd\frac{\mue^2}{\ndraw_j} \frac{(\sum_{i \ne j}\frac{1}{V_i})^2}{(1 + V_j \sum_{i \ne j} \frac{1}{V_i})^2}
\end{align*}

The last component, involving $v_{ji}$ becomes: 
\begin{align*}
\ndraw_j \cd \sum_{i \ne j}\vm_{ji}^2 \cd \p{\frac{\mue}{\ndraw_i} + \var} &= \ndraw_j \cd \sum_{i \ne j}\p{\frac{1}{V_i}\cd \frac{\frac{\mue}{\ndraw_j}}{1 + V_j \sum_{k \ne j}\frac{1}{V_k}}}^2\cd V_i \\
&=\frac{\mue^2}{\ndraw_j} \cd \sum_{i \ne j}\frac{1}{V_i}\cd \frac{1}{(1 + V_j \sum_{k \ne j}\frac{1}{V_k})^2} \\
&= \frac{\mue^2}{\ndraw_j} \cd \frac{\sum_{i \ne j} \frac{1}{V_i}}{(1 + V_j \cd \sum_{k\ne j} \frac{1}{V_k})^2}
\end{align*}
Taken together, these components become: 
$$\mue \cd \frac{(1 + \var \sum_{i\ne j}\frac{1}{V_i})^2}{(1 + V_j\sum_{i\ne j}\frac{1}{V_i})^2} +  \var \cd\frac{\mue^2}{\ndraw_j} \frac{(\sum_{i \ne j}\frac{1}{V_i})^2}{(1 + V_j \sum_{i \ne j} \frac{1}{V_i})^2}$$
$$ + \frac{\mue^2}{\ndraw_j} \cd \frac{\sum_{i \ne j} \frac{1}{V_i}}{(1 + V_j \cd \sum_{k\ne j} \frac{1}{V_k})^2}$$
When we combine them together, we can pull out a common factor to rewrite this as: 
$$\frac{\mue}{(1 + V_j \cd \sum_{i \ne j} \frac{1}{V_i})^2} \cd \p{1 + \var \sum_{i \ne j} \frac{1}{V_i}}^2$$ 
$$ + \frac{\mue}{(1 + V_j \cd \sum_{i \ne j} \frac{1}{V_i})^2} \cd \p{ \var \cd \frac{\mue}{\ndraw_j} \cd \p{\sum_{i \ne j} \frac{1}{V_i}}^2 + \frac{\mue}{\ndraw_j} \cd \sum_{i \ne j} \frac{1}{V_i}}$$ 

We could rewrite this as: 
$$\frac{\mue}{\p{1 + V_j \cd \p{T - \frac{1}{V_j}}}^2} \cd \p{\p{1 + \var \p{T - \frac{1}{V_j}}}}^2 $$
$$+ \var \cd \frac{\mue}{\ndraw_j} \cd \p{T - \frac{1}{V_j}}^2 + \frac{\mue}{\ndraw_j} \cd \p{T - \frac{1}{V_j}}$$
for $T = \sum_{i=1}^{\nplayer} \frac{1}{V_i}$. Next, we begin simplifying, noting that: 
$$\frac{\mue}{\p{1 + V_j \cd \p{T - \frac{1}{V_j}}}^2} = \frac{\mue}{\p{1 + V_j \cd T - 1}^2} = \frac{\mue}{V_j^2 \cd T^2}$$
Simplifying the other coefficient gives: 
\begin{align*}
& \p{1 + \var \p{T - \frac{1}{V_j}}}^2 + \var \cd \frac{\mue}{\ndraw_j} \cd \p{T - \frac{1}{V_j}}^2 + \frac{\mue}{\ndraw_j} \cd \p{T - \frac{1}{V_j}}\\
&=\p{1 + \var \p{T - \frac{1}{V_j}}}^2 + \frac{\mue}{\ndraw_j} \cd \p{T - \frac{1}{V_j}} \cd \p{\var \cd \p{T - \frac{1}{V_j}} + 1}\\
& = \p{1 + \var \p{T - \frac{1}{V_j}}} \cd \p{1 + \var \p{T - \frac{1}{V_j}} + \frac{\mue}{\ndraw_j} \cd \p{T - \frac{1}{V_j}}}\\
& = \p{1 + \var \p{T - \frac{1}{V_j}}} \cd \p{1 +\p{\var +  \frac{\mue}{\ndraw_j} }\cd \p{T - \frac{1}{V_j}}}\\
& = \p{1 + \var \p{T - \frac{1}{V_j}}} \cd \p{1 +V_j\cd \p{T - \frac{1}{V_j}}}\\
& = \p{1 + \var \p{T - \frac{1}{V_j}}} \cd \p{1 +V_j\cd T -1}\\
& = \p{1 + \var \p{T - \frac{1}{V_j}}} \cd V_j\cd T 
\end{align*}
Thus, both terms together simplify down to:   
$$\frac{\mue}{V_j \cd T} \cd  \p{1 + \var \p{T - \frac{1}{V_j}}}$$
Now, that we have simplified the form of the error, we can start looking at sub-proportionality. In this federating group, we are comparing two players $\ns \leq \nlv$ with 
$$V_s = \var + \frac{\mue}{\ns} \quad V_l = \var + \frac{\mue}{\nlv} $$
We will find it useful to rewrite the weighted error of the small player using $T' = T- \frac{1}{V_s} - \frac{1}{V_l}$: 
$$\frac{\mue}{V_s \cd T} \cd  \p{1 + \var \p{T' + \frac{1}{V_l}}}$$
And for the large player: 
$$\frac{\mue}{V_l \cd T} \cd  \p{1 + \var \p{T' + \frac{1}{V_s}}}$$
What we want to be true is that the small player has a lower weighted error, or: 
$$\frac{\mue}{V_s \cd T} \cd  \p{1 + \var \p{T' + \frac{1}{V_l}}}\leq \frac{\mue}{V_l \cd T} \cd  \p{1 + \var \p{T' + \frac{1}{V_s}}}$$
$$\frac{1}{V_s } \cd  \p{1 + \var \p{T' + \frac{1}{V_l}}}\leq \frac{1}{V_l} \cd  \p{1 + \var \p{T' + \frac{1}{V_s}}}$$
$$V_l \cd  \p{1 + \var \p{T' + \frac{1}{V_l}}}\leq V_s \cd  \p{1 + \var \p{T' + \frac{1}{V_s}}}$$
$$V_l+ \var \cd T'\cd V_l + \var \leq V_s+ \var \cd T'\cd V_s +\var$$
$$V_l+ \var \cd T'\cd V_l \leq V_s+ \var \cd T'\cd V_s$$
$$V_l \leq V_s$$
Note that because $V_i = \var + \frac{\mue}{\ndraw_i}$: 
$$\ns \leq \nlv \quad \Rightarrow \quad V_s \geq V_l$$
This means that the overall proof is concluded: using optimal fine-grained federation, sub-proportional error is always satisfied. 
\end{proof}

\egalstab*

\begin{proof}
We will show this result by proving the contrapositive:  if $\alpha >1$, as defined in Definition \ref{def:prop}, then at least one player wishes to leave that coalition for local learning (so stability is violated). 

In our analysis, we will consider two players $\ns < \nlv$, though the players could be arbitrarily situated with respect to $ \frac{\mue}{\var}$. We will assume that all of the players are federating together in coalition $A$. For notational convenience, we will write the federating coalition $A$ of interest as $A = \col \cup \{\nlv\}$, so $\col$ refers to all players except $\nlv$. 

For later portions of this proof, we will find it useful to rewrite the difference in error between the small and large player: 
$$err_S(\col \cup \{\nlv\}) -err_L(\col \cup \{\nlv\})= 2 \var \frac{\nlv - \ns}{\total_\col + \nlv}$$
We note that both terms have the same $\mue$ coefficient, so they differ only in the $\var$ component:

\begin{align*}
&  \var \cd \frac{\sum_{i \in \col}\ndraw_i^2 - \ndraw_s^2 + \nlv^2 + (\total_\col - \ns + \nlv)^2}{(\total_\col + \nlv)^2} - \var \frac{\sum_{i \in \col}\ndraw_i^2 + \total_\col^2}{(\total_\col + \nlv)^2}\\
=&\var \cd \frac{-\ns^2 + \nlv^2 + \total_\col^2 + \ns^2 + \nlv^2 -2 \ns \cd \total_\col + 2\nlv \cd \total_\col-2 \nlv \cd \ns -\total_\col^2}{(\total_\col + \nlv)^2} \\
=&\var \cd \frac{ 2\nlv^2 -2 \ns \cd \total_\col + 2\cd \nlv \cd \total_\col-2 \nlv \cd \ns}{(\total_\col + \nlv)^2} 
\end{align*}
Note that we can rewrite the numerator: 
$$2 (\total_\col + \nlv) \cd (\nlv - \ns) = 2 \cd (\total_\col \cd \nlv - \ns \cd \total_\col + \nlv^2 - \ns \cd \nlv)$$
So the overall difference reduces to: 
$$=\var \cd \frac{ 2 \cd (\nlv -\ns)}{\total_\col + \nlv} $$
as desired. 

Next, we will begin the formal proof. At a high level, what the proof does is first derive conditions where $\alpha > 1$, and then show that, when these conditions are satisfied, at least one player wishes to leave $\col$ for local learning. 

$\alpha > 1$ is implied whenever: 
$$\frac{err_S(\col \cup \{\nlv\})}{err_L(\col \cup \{\nlv\})} \geq \frac{\nlv}{\ns} $$
$$err_S(\col \cup \{\nlv\}) \geq \frac{\nlv}{\ns} \cd err_L(\col \cup \{\nlv\}) $$
(Note that our definition of $\alpha$ is defined with respect to the smaller player.)
Our goal will be to get a lower bound on $\nlv$, given that this occurs. 
Rewriting the subproportionality violation:
$$2 \var \frac{\nlv - \ns}{\total_C + \nlv} \geq \p{\frac{\nlv}{\ns}-1} \cd err_L(C)$$
Plugging in for the value of the large player's error: 
$$2 \var \frac{\nlv - \ns}{\total_C + \nlv} \geq \p{\frac{\nlv}{\ns}-1} \cd \p{\frac{\mue}{\total_C + \nlv} + \var \frac{\sum_{i \in C}\ndraw_i^2 + \total_C^2}{(\total_C+ \nlv)^2}}$$
Cancelling a $\total_\col + \nlv$ in the denominator: 
$$2 \var (\nlv - \ns) \geq \p{\frac{\nlv}{\ns}-1} \cd \p{\mue + \var \frac{\sum_{i \in C}\ndraw_i^2 + \total_C^2}{\total_C+ \nlv}}$$
Rearranging a $\frac{1}{\ns}$ in the right hand side: 
$$2 \var \cd (\nlv - \ns) \geq \p{\nlv-\ns} \cd \p{\frac{\mue}{\ns} + \frac{\var}{\ns} \frac{\sum_{i \in C}\ndraw_i^2 + \total_C^2}{\total_C+ \nlv}}$$
Cancelling $\nlv - \ns$ from both sides: 
$$2 \var \geq \frac{\mue}{\ns} + \frac{\var}{\ns} \frac{\sum_{i \in C}\ndraw_i^2 + \total_C^2}{\total_C+ \nlv}$$
Bringing over common terms: 
$$0\geq \frac{\mue}{\ns}-2 \var  + \frac{\var}{\ns} \frac{\sum_{i \in C}\ndraw_i^2 + \total_C^2}{\total_C+ \nlv}$$
Multiplying by $\total_\col + \nlv$: 
$$0\geq \frac{\mue}{\ns}\cd (\total_C + \nlv)-2 \var\cd (\total_C + \nlv)  + \frac{\var}{\ns} \p{\sum_{i \in C}\ndraw_i^2 + \total_C^2}$$
Collecting terms around $\nlv$: 
$$0\geq \total_C \cd \p{\frac{\mue}{\ns} -2 \var} + \frac{\var}{\ns} \p{\sum_{i \in C}\ndraw_i^2 + \total_C^2} + \p{\frac{\mue}{\ns} - 2 \var} \cd \nlv$$
Note that if $\frac{\mue}{\ns} \geq 2 \var$, then this inequality can never hold (all of the terms on the righthand side are positive, where we need at least one to be negative). So, in order for it to even be possible to have $\alpha \geq 1$, we need $\frac{\mue}{\ns} < 2 \var$, or $\frac{\mue}{2\var} < \ns$. Knowing this, we can pull over terms and flip the sign of the inequality: 
\begin{align*}
\p{2 \var \cd \total_C - \frac{\mue}{\ns} \cd \total_C - \frac{\var}{\ns} \cd \p{\sum_{i \in C} \ndraw_i^2 + \total_C^2}} \geq& \p{\frac{\mue}{\ns} - 2 \var} \cd \nlv\\
\frac{2 \var \cd \total_C - \frac{\mue}{\ns} \cd \total_C - \frac{\var}{\ns} \cd \p{\sum_{i \in C} \ndraw_i^2 + \total_C^2}}{\frac{\mue}{\ns} - 2 \var} \leq& \nlv\\
\frac{-2 \var \cd \total_C + \frac{\mue}{\ns} \cd \total_C + \frac{\var}{\ns} \cd \p{\sum_{i \in C} \ndraw_i^2 + \total_C^2}}{2 \var-\frac{\mue}{\ns} }  \leq& \nlv
\end{align*}

This gives us a lower bound on $\nlv$ for when sub-proportionality is violated. \\
Next, let's look at the condition for when the larger player would choose to leave. 
$$err_L(\col \cup \{\nlv\}) \geq err_L(\{\nlv\})$$
The form of the error in this case is given by: 
$$\frac{\mue}{\total_C + \nlv} + \var \cd \frac{\sum_{i \in C}\ndraw_i^2 + \total_C^2}{(\total_C + \nlv)^2} \geq \frac{\mue}{\nlv}$$
Simplifying: 
\begin{align*}
\var \cd \frac{\sum_{i \in C}\ndraw_i^2 + \total_C^2}{(\total_C + \nlv)^2} \geq& \frac{\mue}{\nlv}-\frac{\mue}{\total_C + \nlv}\\
\var \cd \frac{\sum_{i \in C}\ndraw_i^2 + \total_C^2}{(\total_C + \nlv)^2} \geq& \frac{\mue}{\nlv \cd (\total_C + \nlv)} \cd (\total_C + \nlv - \nlv)\\
\var \cd \frac{\sum_{i \in C}\ndraw_i^2 + \total_C^2}{\total_C + \nlv} \geq& \frac{\mue}{\nlv} \cd \total_C \\
\var \cd \frac{\sum_{i \in C}\ndraw_i^2 + \total_C^2}{\total_C} \geq& \frac{\mue}{\nlv} \cd (\total_C +\nlv)\\
\var \cd \frac{\sum_{i \in C}\ndraw_i^2 + \total_C^2}{\total_C^2} \geq& \frac{\mue}{\nlv\cd \total_C} \cd (\total_C +\nlv)\\
\var \cd \frac{\sum_{i \in C}\ndraw_i^2}{\total_C^2} + \var \geq& \frac{\mue\cd \total_S}{\nlv\cd \total_C}  + \frac{\mue \cd \nlv}{\nlv \cd \total_C}\\
\var \cd \frac{\sum_{i \in C}\ndraw_i^2}{\total_C^2} + \var \geq& \frac{\mue}{\nlv}  + \frac{\mue}{\total_C}\\
\var \cd \frac{\sum_{i \in C}\ndraw_i^2}{\total_C^2} + \var - \frac{\mue}{\total_C} \geq& \frac{\mue}{\nlv}\\
\nlv \cd \p{\var \cd \frac{\sum_{i \in C}\ndraw_i^2}{\total_C^2} + \var - \frac{\mue}{\total_C}} \geq& \mue
\end{align*}
We have a similar situation to before: if this coefficent is negative, then the inequality can never hold (no large player will ever wish to leave). The coefficient is negative when: 
$$\var \cd \frac{\sum_{i \in C}\ndraw_i^2}{\total_C^2} + \var - \frac{\mue}{\total_C}\leq 0$$
We're going to show that this can never happen so long as $\ns > \frac{\mue}{2 \var}$ (necessarily for sub-proportionality to be violated). Rewriting: 
$$\var \cd \sum_{i \in \col} \ndraw_i^2 +\var \cd \total_\col^2 < \mue \cd \total_\col$$
For simplicity, we will rewrite with $\col' = \col \setminus \{\ns\}$. Our term becomes: 
$$\var \cd \sum_{i \in \col'} \ndraw_i^2 +\var \cd  \ns^2+\var \cd (\total_{\col'} + \ns)^2 < \mue \cd \total_{\col'} + \mue \cd \ns $$
Expanding: 
$$\var \cd \sum_{i \in \col'} \ndraw_i^2 +\var \cd  \ns^2+\var \cd \total_{\col'}^2 + \var\cd  \ns^2 + 2 \var \cd \ns \cd \total_{\col'}< \mue \cd \total_{\col'} + \mue \cd \ns $$
$$\var \cd \sum_{i \in \col'} \ndraw_i^2 +\var \cd \total_{\col'}^2 + 2 \var \cd \ns \cd \total_{\col'} + 2 \var \cd  \ns^2< \mue \cd \total_{\col'} + \mue \cd \ns $$
We can show that this never holds by breaking it into two parts. The first one is to show that: 
$$2 \var \cd \ns^2 > \mue \cd \ns $$
which holds by the assumption that $\ns > \frac{\mue}{2\var}$. Next, we can show that
$$\var \cd \sum_{i \in \col'} \ndraw_i^2 +\var \cd \total_{\col'}^2 + 2 \var \cd \ns \cd \total_{\col'} > \mue \cd \total_{\col'} $$
This holds again by the assumption that $\ns > \frac{\mue}{2\var}$. Taken together, this shows that the coefficients on the original equation can never be negative. 

Returning to our original equation: 
$$\nlv \cd \p{\var \cd \frac{\sum_{i \in C}\ndraw_i^2}{\total_C^2} + \var - \frac{\mue}{\total_C}} \geq \mue$$
We know that the coefficient can never be negative, so we can rewrite it as: 
\begin{equation}\label{eq:unstable}
\frac{\mue}{\var \cd \frac{\sum_{i \in \col}\ndraw_i^2}{\total_\col^2} + \var - \frac{\mue}{\total_\col}} \leq \nlv
\end{equation}

As a reminder, our condition for $\alpha \geq 1$ is given by: 
\begin{equation}\label{eq:subpropviol} \frac{-2 \var \cd \total_C + \frac{\mue}{\ns} \cd \total_C + \frac{\var}{\ns} \cd \p{\sum_{i \in C} \ndraw_i^2 + \total_C^2}}{2 \var-\frac{\mue}{\ns} }  \leq \nlv
\end{equation}
What we would like to show is that the bound in Equation \eqref{eq:unstable} is lower than the bound in Equation \eqref{eq:subpropviol}: any situation where the $\nlv$ term is large enough so that $\alpha > 1$, we know that the large player already wishes to leave.

So, what we want to show is: 
$$\frac{\mue}{\var \cd \frac{\sum_{i \in \col}\ndraw_i^2}{\total_\col^2} + \var - \frac{\mue}{\total_\col}} \leq \frac{-2 \var \cd \total_C + \frac{\mue}{\ns} \cd \total_C + \frac{\var}{\ns} \cd \p{\sum_{i \in C} \ndraw_i^2 + \total_C^2}}{2 \var-\frac{\mue}{\ns} } $$
Multiply the top and bottom of the righthand side by $\frac{\ns}{\total_\col}$: 
$$\frac{\mue}{\var \cd \frac{\sum_{i \in \col}\ndraw_i^2 + \total_\col^2}{\total_\col^2} - \frac{\mue}{\total_\col}} \leq \frac{\mue -2 \var \cd \ns + \frac{\var}{\total_\col} \cd \p{\sum_{i \in C} \ndraw_i^2 + \total_C^2}}{2 \var\cd \frac{\ns}{\total_\col}-\frac{\mue}{\total_\col} } $$
Next, we multiply the LHS by $\total_\col^2$ on both top and bottom and the RHS by $\total_\col$ on both top and bottom. 
$$\frac{\mue\cd\total_\col^2}{\var \cd B - \mue \cd \total_\col} \leq \frac{\mue \cd \total_\col -2 \var \cd \ns \cd \total_\col + \var \cd B}{2 \var\cd \ns-\mue } $$
Next, we cross multiply. On the lefthand side, we obtain: 
$$2 \mue \cd \total_\col^2 \cd \var \cd \ns - \mue^2 \cd \total_\col^2$$
On the righthand side: 
$$\mue \cd \total_\col \cd \var \cd B - \mue^2 \cd \total_\col^2 - 2 (\var)^2 \cd \ns \cd \total_\col \cd B + 2 \mue \cd \var \cd \ns \cd \total_\col^2 + (\var)^2 \cd B^2 - \mue \cd \var \cd B \cd \total_\col$$
$$= - \mue^2 \cd \total_\col^2 - 2 (\var)^2 \cd \ns \cd \total_\col \cd B + 2 \mue \cd \var \cd \ns \cd \total_\col^2 + (\var)^2 \cd B^2$$
Bringing the two sides together, what we wish to show is: 
$$
2\mue \cd \total_\col^2 \cd \var \cd \ns - \mue^2 \cd \total_\col^2$$
$$\leq - \mue^2 \cd \total_\col^2 - 2 (\var)^2 \cd \ns \cd \total_\col \cd B + 2 \mue \cd \var \cd \ns \cd \total_\col^2 + (\var)^2 \cd B^2$$
or: 
$$0\leq   - 2 (\var)^2 \cd \ns \cd \total_\col \cd B + (\var)^2 \cd B^2$$
Dropping the common (positive) terms of $B, (\var)^2$, gives: 
\begin{align*}
0\leq&  - 2\cd \ns \cd \total_\col +  B\\
  2\cd \ns \cd \total_\col \leq &  \sum_{i \in \col}\ndraw_i^2 + \total_\col^2
\end{align*}

We'll use the same $\col' = \col \setminus \{\ns\}$ rewriting as before. This gives us: 

\begin{align*}
2 \cd \ns \cd \total_{\col'} + 2 \cd \ns^2 \leq & \sum_{i \in \col'}\ndraw_i^2 + \ns^2 + (\total_{\col'} + \ns)^2\\
2 \cd \ns \cd \total_{\col'} + 2 \cd \ns^2 \leq& \sum_{i \in \col'}\ndraw_i^2 + \ns^2 + \total_{\col'}^2 + \ns^2 + 2 \ns \cd \total_\col'
\end{align*}
Cancelling: 
$$0 \leq  \sum_{i \in \col'}\ndraw_i^2 + \total_{\col'}^2 $$
This is always satisfied. 
This tells us that the inequality is satisfied and is strict whenever $\col' \ne \emptyset$. If $\col' = \emptyset$, then the bound for each becomes: 
$$\frac{\mue}{2 \var - \frac{\mue}{\ns}} = \nlv$$
If this equation holds exactly, the exact proportionality holds and the large player gets equal error by federating with $\ns$ or doing local learning. 
\end{proof}

\end{document}
\endinput